\pdfoutput=1
\documentclass[sigconf, nonacm]{acmart}
\usepackage{graphicx}
\usepackage{algorithm}
\usepackage{amsthm}
\usepackage[english]{babel}
\usepackage{amsthm}
\usepackage{url}
\usepackage{color}
\usepackage{algorithmicx}
\usepackage{subfigure}
\usepackage{multirow} % Required for multirows
\usepackage[noend]{algpseudocode}
\newcommand\vldbdoi{XX.XX/XXX.XX}
\newcommand\vldbpages{XXX-XXX}
\newcommand\vldbvolume{14}
\newcommand\vldbissue{1}
\newcommand\vldbyear{2020}
\newcommand\vldbauthors{\authors}
\newcommand\vldbtitle{\shorttitle} 
\newcommand\vldbavailabilityurl{URL_TO_YOUR_ARTIFACTS}
\newcommand\vldbpagestyle{plain}

\begin{document}

% ****************** TITLE ****************************************

\title{One stone, two birds: A lightweight multidimensional  learned index with cardinality support}

% possible, but not really needed or used for PVLDB:
%\subtitle{[Extended Abstract]
%\titlenote{A full version of this paper is available as\textit{Author's Guide to Preparing ACM SIG Proceedings Using \LaTeX$2_\epsilon$\ and BibTeX} at \texttt{www.acm.org/eaddress.htm}}}

% ****************** AUTHORS **************************************


\author{Yingze Li}
\affiliation{%
\institution{Harbin Institude of Technology}
\city{Harbin}
\state{China} 
}
\email{1190202126@stu.hit.edu.cn}

\author{Hongzhi Wang }
\affiliation{%
\institution{Harbin Institude of Technology}
\city{Harbin}
\state{China} 
}
\email{	wangzh@hit.edu.cn}

\author{Xianglong Liu }
\affiliation{%
\institution{Harbin Institude of Technology}
\city{Harbin}
\state{China} 
}
\email{1190200525@stu.hit.edu.cn}

\begin{abstract}

Innovative learning-based structures have recently been proposed to tackle index and cardinality estimation(CE) tasks, specifically learned indexes and data-driven cardinality estimators. These structures exhibit excellent performance in capturing data distribution, making them promising for integration into AI-driven database kernels. However, accurate estimation for corner case queries requires a large number of network parameters, resulting in higher computing resources on expensive GPUs and more storage overhead. Additionally, the separate implementation for CE and learned index result in a redundancy waste  by storage of single table distribution twice. These present challenges for designing AI-driven database kernels. As in real database scenarios, a compact kernel is necessary to process queries within a limited storage and time budget. Directly integrating these two AI approaches would result in a heavy and complex kernel due to a large number of network parameters and repeated storage of data distribution parameters.  \looseness=-1

Our proposed CardIndex structure effectively killed two birds with one stone. It is a fast multi-dim learned index that also serves as a lightweight cardinality estimator with parameters scaled at the KB level. Due to its special structure and small parameter size, it can obtain both cumulative density function (CDF) and probability density function (PDF) information for tuples with an incredibly low latency of $1-10 \mu s$. For tasks with low selectivity estimation, we did not increase the model's parameters to obtain fine-grained point density. Instead, we fully utilized our structure's characteristics  and proposed a hybrid estimation algorithm in providing   fast and exact results. Extensive experiments show that our CardIndex outperforms the state-of-the-art CE methods by 1.3-114$ \times $ on low  selectivity queries, while its inference latency on CPU and storage consumption is also 1-2 orders of magnitude smaller than these deep models on GPU. On index tasks, our CardIndex's index performance is $30\%$-$40\%$ faster in point queries and 4-10$ \times $ faster in range queries compared to traditional multi-dim indexes.

\end{abstract}

\maketitle

\pagestyle{\vldbpagestyle}
\begingroup\small\noindent\raggedright\textbf{PVLDB Reference Format:}\\
\vldbauthors. \vldbtitle. PVLDB, \vldbvolume(\vldbissue): \vldbpages, \vldbyear.\\
\href{https://doi.org/\vldbdoi}{doi:\vldbdoi}
\endgroup
\begingroup
\renewcommand\thefootnote{}\footnote{\noindent
This work is licensed under the Creative Commons BY-NC-ND 4.0 International License. Visit \url{https://creativecommons.org/licenses/by-nc-nd/4.0/} to view a copy of this license. For any use beyond those covered by this license, obtain permission by emailing \href{mailto:info@vldb.org}{info@vldb.org}. Copyright is held by the owner/author(s). Publication rights licensed to the VLDB Endowment. \\
\raggedright Proceedings of the VLDB Endowment, Vol. \vldbvolume, No. \vldbissue\ %
ISSN 2150-8097. \\
\href{https://doi.org/\vldbdoi}{doi:\vldbdoi} \\
}\addtocounter{footnote}{-1}\endgroup

\ifdefempty{\vldbavailabilityurl}{}{
\vspace{.3cm}
\begingroup\small\noindent\raggedright\textbf{PVLDB Artifact Availability:}\\
The source code, data, and/or other artifacts have been made available at \url{https://github.com/LiYingZe/CardIndex}.
\endgroup
}

\section{Introduction}

Learning-based architectures are offering novel prospects to traditional database problems~\cite{li2021ai}. The learned index methodology~\cite{kraska2018case} endeavors to acquire a function that associates the key of an indexed tuple with its cumulative density function (CDF), while the data-driven cardinality estimators~\cite{yang2019deep} aim to sample from the probability density function (PDF) approximated by deep autoregressive models. Due to their excellent performance in estimation quality and rapid index searching, they are expected to be integrated into the AI-driven database kernel~\cite{jasny2020db4ml}, creating a new generation of database products.

Even though these mechanisms give promising chances, embedding both of them into a single database brings challenges.
In practical databases, a compact and elegant database kernel is necessary for the efficient operation of the database system~\cite{paul1987architecture}. That is, query processing is required to be accomplished within a limited storage space and time budget. Although these learning-based strategies are excellent in learning data distribution, the existing methods are often too cumbersome to be applicable in the scenario of multidimensional data. For example, given a table with 80 million tuples. The autoregressive network proposed by Naru~\cite{yang2019deep} requires 10 hours of training, more than  100 MB  of space to store parameters, and at least $ 20ms $ of inference latency on an expensive GPU. These heavy structures demand  many training resources for a converged model, many computing resources for inference, and large memory overhead to store the distribution of tables.\looseness=-1

This dilemma  is difficult to solve. The reason is that, for learning-based tasks, if we want to obtain an accurate estimation of some corner cases, e.g.   some queries with small cardinalities, then we should pay the corresponding price. The  model parameter scale is increased to achieve higher prediction accuracy. Larger model parameters require more computing resources to train, higher computing latency to inference, and more space  to store. Meanwhile, considering that multi-dimensional learned indexes and CE tasks   learn  from the distribution of the same table, the information of a single table is stored twice via two different networks, which leads to redundant waste of space.\looseness=-1

%讨论为什么learned index和cardinality estimator可以放到一起做
Fortunately, this gives us chance to solve both indexing and CE problems for multi-dimensional scenarios      together with one lightweight structure. That is, a small model is used to learn the distribution of a single table and to act as both a learned index and a cardinality estimator. The parameters of the model are controlled at the KB level to achieve the $ \mu s $ level of inference latency on a cheap CPU than an expensive GPU. Given its tiny parameter size, it cannot achieve very accurate point density estimation. This causes poor performance in directly estimating some small cardinality queries, and it needs multiple auxiliary refining structures to accurately pinpoint the index address. However, for these low selectivity queries, it is much faster to use index execution directly to obtain accurate cardinality values compared to those sampling methods. For example, for queries with a result size of 100-1000, it only takes about $50 \mu s $ to search the precise cardinality directly via the index on CPU, and $ 20$-$200ms  $ to obtain an approximate result for a progressive sampling using autoregressive models on GPU.

Above analysis inspires us to construct a general lightweight structure for these two tasks of index learning and CE, which not only plays as an efficient multidimensional index but also stands as a lightweight cardinality estimator. It can efficiently identify the query size and select the appropriate estimation strategy. For  estimation over high selectivity, it uses the progressive sampling strategy for estimation. For small selectivity estimation tasks, it directly uses the index structure to obtain the exact cardinality value in less time.

\begin{figure}
\centering
\includegraphics[height=3cm]{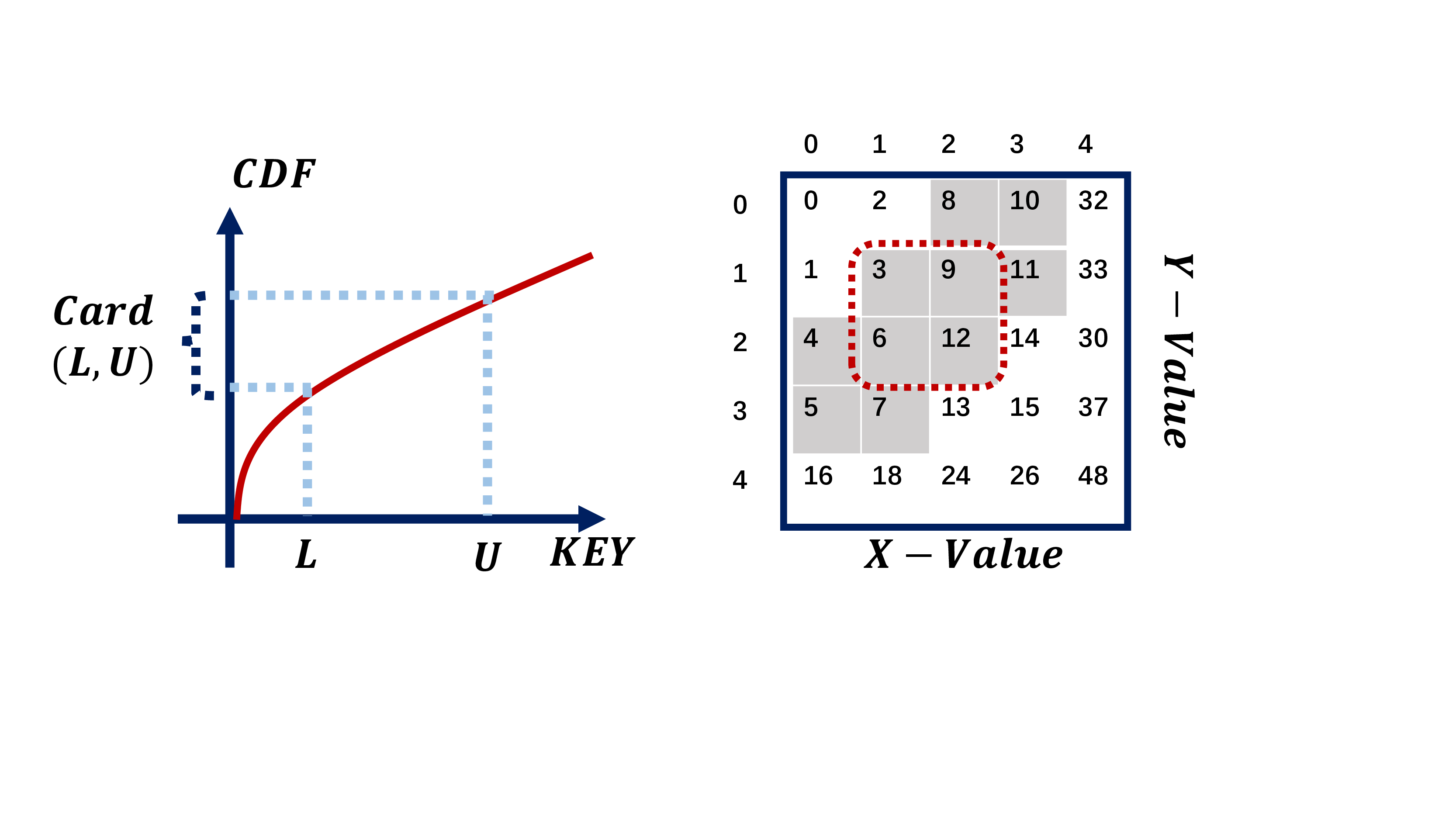}
\caption{ In a single dimension, a simple subtraction is sufficient for transforming an index into a cardinality estimator. However, multi-dim cases are more difficult. Directly using subtraction in Z-Order will result in a larger estimation error.  }
\label{F1}
\end{figure}

The discussed structure is intuitive for one-dimensional data, but direct implementations are not scalable to multi-dimensional situations. Ji et al~\cite{ji2022performance}.  modified the learned index to achieve CDF estimation on both sides of the one-dim query through twice point queries. By subtracting, the estimation of the one-dimensional cardinality is obtained. However, these efforts are not suitable for multi-dimensional cases, where the problem is much more complex. As demonstrated in Fig~\ref{F1}, using the Z-Order curve directly as a dimensionality reduction bridge leads to greater estimation errors, as not all records between $Z_{min}$ and $Z_{max}$ will be contained in the query box. Therefore, the barrier exists in utilizing multi-dimensional CDF in CE tasks.\looseness=-1
%此处增加1句讨论作为结论，引出下文

To achieve the goal on multi-dimensional data, we crossed the barrier. Our solution derives a unified structure that can obtain CDF and PDF information simultaneously under only a single inference step. CDF information can be used for index queries, while PDF information can be used for progressive sampling estimation. Therefore, our proposed CardIndex structure effectively killed two birds with one stone. In the case of multidimensional index tasks, it can be utilized as a learned index to swiftly conduct point queries with an average latency of $1\mu s$  and the performance of range query is not inferior to the traditional structure. In the case of cardinality estimation, only a few computations on the CPU are required to achieve close estimations with a similar latency compared to mainstream deep estimators on GPU. Moreover, our proposed hybrid estimation  method  brought us a chance to get the precise result of low selectivity queries. Specifically, our proposed CardIndex structure faces the following challenges: \looseness=-1

\textbf{C1: Data organization:}  How can multidimensional data be organized so that multidimensional CDFs can be associated with sorted data addresses to efficiently support multidimensional queries and estimation tasks?

\textbf{C2: Model Implementation:} How can a learning-based model be implemented to replace traditional database index structures, such as KDB tree and R tree, while maintaining low latency? 

\textbf{C3: Algorithm design:}  What methodology can be utilized in inferring this unified structure to support both index processing and cardinality estimation? Furthermore, how can the precision quality of the model be preserved in problem resolution, particularly when dealing with limited model parameters?

We have addressed these challenges, and given the following contributions as follows:

1. We propose the CardIndex structure, a lightweight multidimensional  learned index that learns the Z-Order bit encoding distribution of data through the autoregressive model. The learned parameters of our structure are in KB scale. Thus, our CardIndex can perform one fast inference under $ 1-10\mu s $ to output the  CDF and PDF information required by Index and CE. We designed algorithms for point query, range query, and CE tasks, making this the first attempt to solve both CE and multidimensional indexing tasks in a single learned structure.

2. For the indexing task, we developed two approaches to improve the performance, Pre-$ k $ neuron linkage and Linear Refinement. The former effectively shortens the neuronal dependency complexity, making the structure more scalable on Z-Ordered bit inference length. The latter refines the approximated CDF estimation of deep structure through cheaply linear functions. Thus we could obtain a more accurate address through the index, which makes query processing more efficient.

3. For the cardinality estimation task, we propose the hybrid estimation method. We first use several point index queries to probe the upper bound of the cardinality selectivity, namely, our FastCDFEst algorithm. This helps us distinguish whether the   cardinality of the query is large or small. We only use progressive sampling of our deep model to estimate large cardinality, while for the estimation of small cardinality, we directly use the multidimensional index part to achieve precise estimation in much less time.

4. We compared our cardinality estimation algorithm with the state-of-the-art single table estimators and multidimensional indexes on multiple workloads of real and synthetic data. Experiment results show that our hybrid estimation algorithm improves estimation performance by  1.3-114$  \times  $  in scenarios with extremely small cardinality queries. The model size and inference time of our method on CPU is also 1-2 orders of magnitude smaller than current deep estimators on GPU. Experiments also show that our index is 30-40\% faster on point query tasks and  4-10$ \times $ faster than traditional architectures under certain workloads un range queries.

The remaining sections of this paper are structured as follows. Section 2 reviews related works. Section 3 introduces the fundamental concepts of our approach. Section 4 outlines the underlying structure of our index, and Section 5 provides a detailed description of the algorithm used for estimation and index queries utilizing our CardIndex. In Section 6, we demonstrate the effectiveness of our methodology through extensive experimentation. We conclude our work in Section 7.  \looseness=-1

\section{RELATED WORK}

\label{Sec2}

\textbf{Z-Order Curve}: A space-filling curve that maps multidimensional data to one-dimensional~\cite{tropf1981multidimensional}. It is also good in locality~\cite{dai2003locality}, meaning that the adjacent Z-Order value tuples tend  to be similar in many attribute values. As a result, it is widely used in the implementation of the  multidimensional index, such as UB-Tree\cite{ramsak2000integrating}. The easiest approach to using the Z-Order curve in a range query is to calculate the Z-values of the top left $ Z_{min} $ and bottom right $ Z_{max} $ endpoints of the query box and scan for all the data between them~\cite{8788832}. However, many elements between $ Z_{min} $ and $ Z_{max} $ are not in the query box. Therefore, discontinuous gaps arise  between $ Z_{min} $ and $ Z_{max} $. To skip these gaps, tropf et al~\cite{tropf1981multidimensional}  proposed the   \textit{getBIGMIN}  method. Given a tuple $ t $ outside the query box,  \textit{getBIGMIN}  shall ``jump" to the minimum Z-value greater than $ t $ in the query box within $ O (n) $ time where $ n $ is the bit encoding length. The counterpart of the above method is the  \textit{getLITMAX}  method, which can get the maximum Z-value in the query box less than the tuple $ t $ in $ O (n) $ time.   \looseness=-1

\textbf{Learned Index}: By using a learned function that maps a search key to its address, the learned index achieves point query in constant time. Compared with traditional index structures, this approach avoids searching in a log-scale depth tree index. Under one-dimensional cases, RMI~\cite{kraska2018case} first utilized this idea by using MLPs to map the search key to its CDF. After RMI, many newly learned indexes sprang up. For example, the Fitting-Tree~\cite{galakatos2019fiting} structure fits data with piecewise linear functions and gives bound guarantees for data coordinates. Meanwhile, a lot of work has been done on learned index in multidimensional scenarios. Tsunami~\cite{ding2020tsunami} takes the organization structure of grid files and uses the tuned grid partitioning to organize the index. RSMI~\cite{qi2020effectively} and ZM~\cite{8788832} are similar to our approach of indexing multidimensional data under the Z-Order space-filling curve, but it does not take advantage of Z-Order’s optimization to skip the discontinuous regions.

\textbf{Data-driven CE}: Through the data distribution learned from the model, we can estimate the  cardinality of the query. There are two kinds of representative   genres in data-driven methods:  (1) SPN (Sum Product Network)~\cite{hilprecht2019deepdb}:  SPN recursively performs horizontal and vertical decomposition of the table through a KD tree-like space partition strategy. It uses sum nodes to combine different rows as clusters. The multiplication node is used to decompose the attribute group with weak column correlation. (2) Autoregressive model~\cite{yang2019deep}: Through the strong fitting ability of deep learning networks, deep autoregressive models learns the joint distribution between data attributes, thus achieving good point density estimation. The representative method is Naru, which achieves a good estimation of query distribution through a progressive sampling algorithm. This algorithm belongs to a special branch of Monte Carlo integrated sampling, which makes the sampled model concentrate on some attribute columns with higher density and saves the sampling times.

%(1) SPN(Sum Product Network): Through KD tree like space partition strategy, SPN recursively performs horizontal and vertical decomposition of the table. It use sum nodes in combining different rows as clusters. The multiplication node is used to decompose the attribute group with weak column correlation. (2) Autoregressive model: Through the strong fitting ability of deep learning network, deep autoregressive model remembers the joint distribution between data attributes, thus achieving good point density estimation.  The representative method is Naru, which achieves good estimation of query distribution through progressive sampling algorithm. This algorithm belongs to a special branch of Monte Carlo integrated sampling, which makes the sampled model concentrate moon some attribute columns with higher density and saves the sampling times.

\section{PRELIMINARIES }

In this section, we shall commence by presenting fundamental notions outlined in \S~\ref{Sec3.1}. Subsequently, we shall derive the key concepts of indexing and estimation by utilizing the aforementioned fundamental notions in conjunction with the bit joint distribution discussed in \S~\ref{sec 3.2}. Lastly, we will deliberate on the insights about the information our model has learned from both the cumulative distribution function (CDF) and probability density function (PDF) perspectives, as called upon in \S~\ref{sec 3.3}.

\label{Sec3}
\subsection{Basic Concepts}

\label{Sec3.1}
For a relational table $ T $ on $ N $ tuples, and $ m $ attributes $ \{     A_1, A_2 \dots A_m \} $, each of which is encoded by $ b $ binary bits, i.e. $ A_i = D_{i1} D_{i2} \dots D_{ib} $. For a query predicate: $ \theta  :A_1 \times A_2 \dots   \times A_m \rightarrow \{0,1\} $, tuples satisfying predicate  in table $ T  $ forms a result set $ R_0=\{ t \in T:\theta(t)=1\} $.  {The index demands an auxiliary structure in accelerating the process of getting the result set $ R_0  $.  And CE requires the cardinality estimation of the result set, i.e., $ card(\theta) = | R_0 | $.}

The joint distribution $ P(t)=o(t)/N $ of the tuple $ t $ is related to the query selectivity $ sel(\theta) $, where $ o(t)  $ is the number of occurrences of tuple $ t $ in table $ T $. Therefore,

\begin{equation}
sel(\theta) = \sum_{t\in A_1 \times A_2 \times \dots A_m} \theta(t) \cdot  P(t)
\end{equation}

Meanwhile, if data ordering satisfies the bit ordering property $ L= \left\{ D_{x1}, D_{x2}... D_{xb}  \right\}   $,  ordering relationship  $ t_2 < t_1  $ is satisfied, if and only if: $  \exists k: s.t. t_2.D_{xk} < t_1.D_{xk}  $ and $ \forall j < k: t_2. D_{xj} = t_1. D_{xj} $. In this particular ordering, the tuple address distribution $ CDF(t) $ is also closely related to the   distribution of tuples $ P(t) $ :\looseness=-1

\begin{equation}
CDF(t)= \sum_{t_i <  t} o(t_i)/ N =  \sum_{t_i <  t}  P(t_i)
\end{equation}

$ L_{Lex} $ is lexicographical ordering, when $ L_{Lex}= \{ D_{11}, D_{12} , \dots D_{1b},$  $  D_{21}, D_{22}, \dots D_{2b} \dots D_{m1}, D_{m2},\dots D_{mb} \}   $.

$  L_Z $ is Bitwise entanglement ordering(Z-Order), when $ L_{Z}=  \{D_{11},\\ D_{21}, \dots D_{m1}, D_{12}, D_{22}, \dots D_{m2} \dots D_{1b}, D_{2b},\dots D_{mb} \}  $.

To process multidimensional range index queries more efficiently, we convention the  bit ordering as Z-Order $ L_Z $   and adopt the following notations: We   use $ n $ to replace $mb $ in denoting the total coding length, set the ordering notation as  $ L_Z = \left\{ Z_1,Z_2,\dots Z_{n} \right\}   $, tuple $ t   $ under Z-Ordering is denoted as $ t =(z_1,z_2\dots z_{n} )$. 

\subsection{Bitwise Joint Distribution}

\label{sec 3.2}
We note that for a given Z-Ordering $ L_Z$, the tuple distribution $ P(Z_{ 1},Z_{ 2}... Z_{ n}) $, can naturally take the following decomposition: $  P(Z_{ 1},... Z_{ n}) = P(Z_{ 1}) \cdot   P(Z_{ 2}|Z_{ 1}) \dots  P(Z_{ n}|Z_{ 1},Z_{ 2},\dots Z_{ n-1})$.

Therefore, it is possible for us to use the deep autoregressive model to estimate the above decomposed  probability function. We will set the chain of autoregressive reasoning as the Z-Order $ L_Z $, in predicting:$ \{  \hat{P} (Z_{ i}|Z_{ 1},Z_{ 2}, \dots Z_{i-1}) \} $.  {That is, the  model takes the given Z-encoded tuple: $ t =  (z_{1},z_{2},\dots z_{ n}) $ as the input, and gets $   n $ conditional probability estimations:$ \{ \hat{P} (z_{ i}|z_{ 1},z_{ 2}, \dots z_{i-1}) \} $ as the output.}

Based on above input and output of the model, we use the cross entropy loss in calculating the distance between the predicted and ground truth Z-Ordered binary sequences:
$ H(P,\hat{P}) = -\sum_{t\in T} P(t)log(\hat{P}(t)) $. Gradient descent is used to train the model. \looseness=-1

Through the learned Z-Order distribution, we implement the general idea for point queries. Since the point index query of tuple $t = ( z_1,z_2\dots z_{n} )$ needs to estimate all the tuples that come ahead of it: $\sum_{t_{i}  <  t}  o(t_i)$. Directly traversal of  each element in table $ T $ to calculate   the above sum  is  unwise, which has a complexity is $ O(N) $, where $N$ is the cardinality of the entire table.  

For any table $ T $ sorted by Z-Ordering, any tuple $ t_y $ smaller than $ t_x  $ must satisfy that $ \exists i ,  t_y.z_{i}=0 ,  t_x.z_{i}=1 $ and $ \forall j < i: t_y. z_{j} = t_x. z_{j} $.
Therefore, to index tuple $t_x$, all records in table $T$ smaller than $t_x$ can be virtually regrouped based on the longest common bit string length starting from bit position 0. Specifically, let $G(t_x,i) = \{ t_y \in T , \forall j<i, t_y.z_j = t_x.z_j  ,  t_y.z_i < t_x.z_i \}$.
As a result, the CDF of $t_x$ can be calculated according to the formula: $ CDF(t_x) = \sum_{i\leq n} {\left| G\left(t_x,i\right)\right| } /{N}$. Note that the conditional probabilities learned by our model correspond exactly to the cardinals of these groups, where $ Pr(Z_{i}=0|Z_{1}, \dots, Z_{i-1} = z_1,\dots, z_{ i-1}) \times z_{ i}$ corresponds to $ {|G(t_x,i)|}/{N} $. Thus, we can use the following conditional probability sum expression in  obtaining the index address distribution. \looseness=-1

\begin{equation}
CDF(t) = \sum_{i\leq n} Pr(Z_{ i}=0|Z_{1} \dots  {Z_{i-1}} = z_1\dots z_{ i-1}) \cdot z_{ i}
\end{equation}

Considering the accuracy of the model and the computational efficiency, we truncate the above summation and use the following approximation to calculate only the first $ \ell $ terms of the sum formula:

\begin{equation}
CDF(t) \simeq \sum_{i \leq \ell} Pr(Z_{ i}=0|Z_{1} \dots  {Z_{i-1}} = z_1\dots z_{ i-1}) \cdot z_{ i}
\end{equation}

This sum can be computed in $ O(\ell) $ time, where $ \ell $ is a fixed constant that represents the first $ \ell $ representative bits. For example, we take the first 32 bits of Z-Order encoding in our experiments.

\subsection{{Insights From Index And CE Perspective}}

\label{sec 3.3}
The distribution acquired by our model through Z-Order can be modeled as a probabilistic transition tree, which can be visualized in Figure~\ref{F2}. Under the Z-Order bit sequence, the output of our autoregressive model is equivalent to a probabilistic transition path from the root node to the leaf node in the tree. These paths can be invoked in two ways, depending on the usage requirements:

\begin{figure}
\centering
\includegraphics[width=9cm]{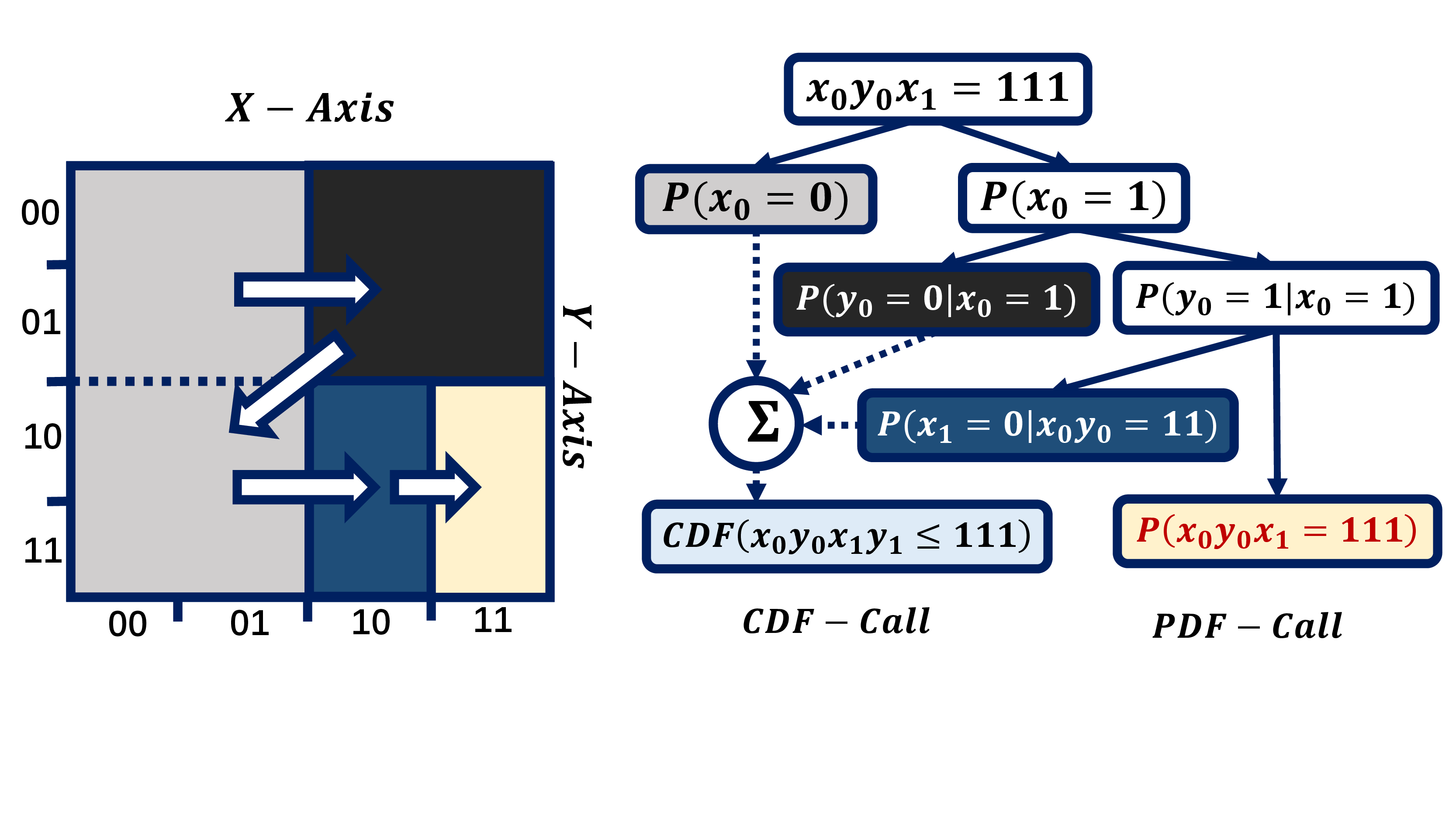}
\caption{What Our structure learns}
\label{F2}
\end{figure}

\textbf{1. CDF Call}: This method returns the summation of probabilities for all child nodes' sibling nodes with a bit value of 1.

\textbf{2. PDF Call}: This method directly outputs the probability density value that corresponds to the end of the Z-Order path.

From the perspective of the multidimensional index, the above-mentioned CDF call can be perceived as a traversal of the Z-Ordered space partitioning tree. Each $ w $ bit output by the model is equivalent to a visit to an internal node of the KDB tree. By summing up the probability value corresponding to these $ w $ bits, the model can directly predict the corresponding position of the leaf node.

From the perspective of the cardinality estimator, the above-mentioned PDF call achieves the acquisition of point density estimation. By employing some existing sampling techniques, such as progressive sampling, we can effectively utilize this point density estimation information for the cardinality estimation of the query.

\section{CardIndex Structure}

In this section, we propose the basic hierarchical structure of our model, as illustrated in Fig~\ref{F3}. At the highest level, a nimble Z-Ordered autoregressive network is utilized to perform two sub-tasks, namely indexing and estimation, achieved by invoking CDF and PDF calls. The remaining nodes refine the CDF value produced by the root using linear functions to achieve precise address information. We discuss the layout and optimization techniques of our model in \S~\ref{Sec 4.1}  and introduce our implementation details in \S~\ref{Sec 4.2}. 

\begin{figure}
\centering
\includegraphics[width=9cm]{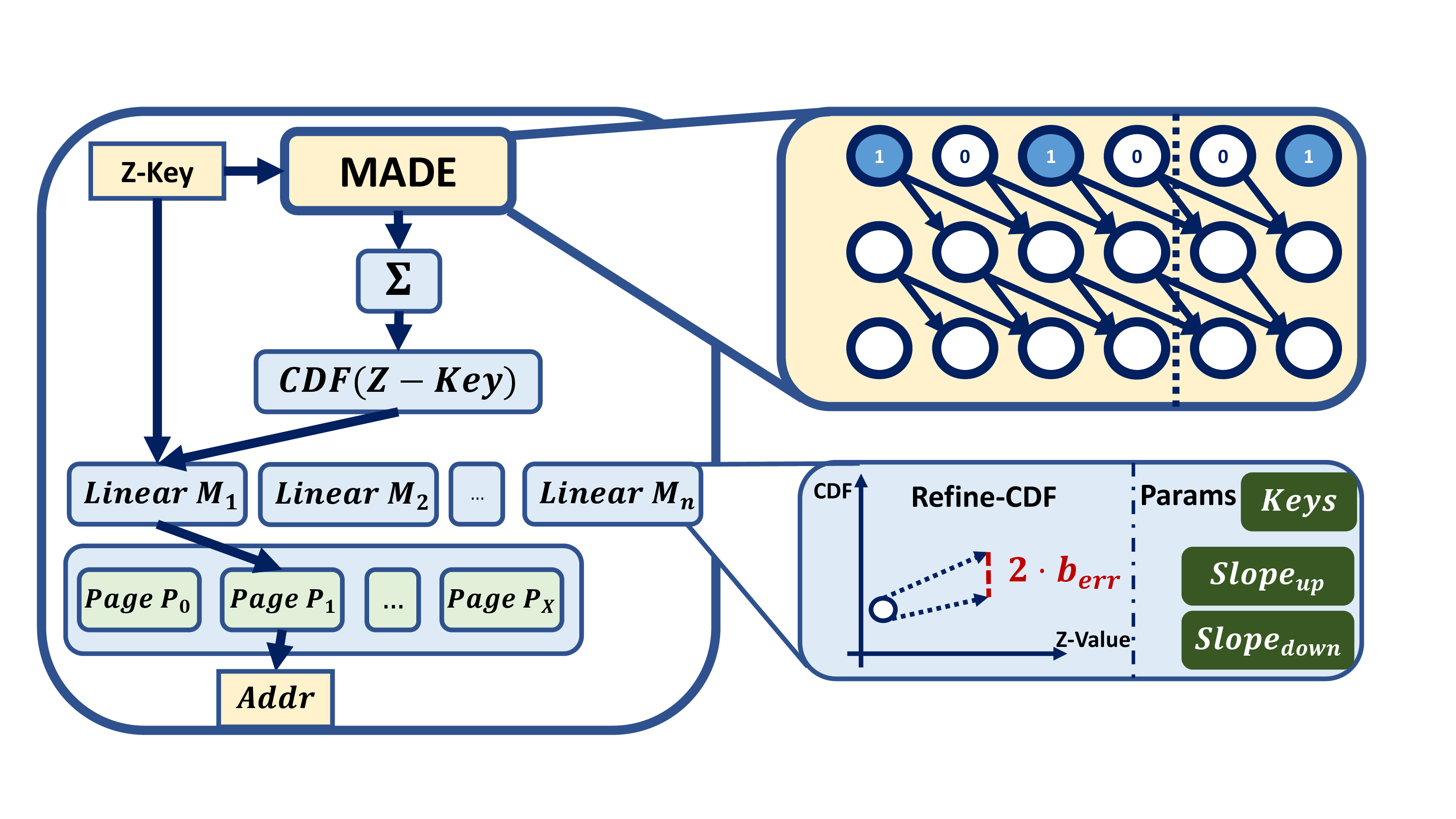}
\caption{ The layout of CardIndex }
\label{F3}
\end{figure}

\subsection{Layout}
\label{Sec 4.1}

In principle, any deep autoregressive model can serve as the CDF estimate structure for the root. However, for the sake of efficient inference, we employ a three-layered   MADE~\cite{germain2015made} as our autoregressive structure.  Given that   both index and CE tasks require low inference latency, and the index task  requires a precise tuple address of a given query, we employ two optimization approaches from the aspect of shortening the dependency complexity of the deep model and finding cheaper functions to decrease the estimated CDF  error.

\textbf{Pre-$ k $ neuron linkage:} To achieve scalability of binary coding length, we use limited mask connection of neurons to establish the conditional probability information fitting between bits. In detail, we adopt the following conventions for mask settings. The $i$-th neuron of the current layer only depends on the output of the $ (i-k) $-th to $ (i-1) $-th neurons of the previous layer. The upper-right corner of Fig~\ref{F3} illustrates the case when $ k=2 $. This method yields better scalability on input coding length. In contrast, for the original MADE network, the $i$-th neuron of the current layer depends, in the worst case, on the output of all the neurons of the first $ (i-1) $-th output of the previous layer. This results in a complexity of scale $O(n^2)$ for the calculation of the network between layers, where $n$ is the length of the encoded bits. Our approach results in a computational complexity of each layer being $O(n\times k)$. This feature enables us to calculate density functions quicker and   index data faster. Although this setting may reduce the performance of the model in  point density estimation to some extent, it will not affect the cardinality estimation performance of our structure too much.  This is because estimations that engage with large cardinalities do not require particularly fine-grained point density estimations. For small cardinality estimations, as we will see in \S~\ref{sec:CE}, can be delegated to the index rather than the inaccurate sampler to get precise cardinality values efficiently. \looseness=-1

%When the mask connection of neurons is used to establish the conditional probability information between bits, in order to achieve the extensibility of the encoding length, we adopt the following conventions for mask Settings: The $ i- th  $ neuron of the current layer only depend on the output of the $ (i-K) th $  to $ (i-1) th $ neurons of the previous layer, so as to better obtain the scalability of coding length. A corresponding example is demonstrated in the upper left of Figure \ref{F3}. In contrast, for the original MADE network, the $ i- th  $ neurons of the current layer are, in the worst case, dependent on the output of all the neurons of the first $ (i-K) th $ output of the previous layer, the calculation of the network between layers will introduce a complexity of scale $ O(n^2) $, where $ n $ is the length of the encoded bits. With our approach, the computational complexity of each layer is $ O(n*w) $, so the CDF estimate can be calculated more quickly to index the data. Although such setting may reduce the performance of the model in distribution estimation to some extent, we found in \S~\ref{sec5.4} of the experiment that there are some specific parameter Settings that can balance the two directions of high efficiency and high quality estimation.

\textbf{Linear Refinement:} After using the deep network to nonlinearly fit the  data distribution of first $\ell$ bits, we use a more cheaply calculated linear function to refine the CDF of the data. We  take the fully Z-Order coded key as the input and the CDF address of the index as the refined output. Our linear interpolation   process is similar to the construction of the works of FITing tree~\cite{galakatos2019fiting}. By maintaining the maximum error limit $ b_{err} $ for linear interpolation when building the tree, two sets of linear parameters are obtained to ensure the refined CDF of each tuple will be sandwiched between a region with the maximum estimated error within $ b_{err} $. 
In deployment, given the stored starting point and its two linear boundary slopes, we can bound the CDF value of the  required key in an interval predicted from these two functions satisfying that the interval length does not exceed $ 2\times b_{err }$. 
%\textcolor{red}{For example, we can store the starting coordinates of the linear function $(x_0,y_0)$ and the upper and lower slopes $l_1,l_2$, we are able to bound the expected CDF value of the key $ x_1 $ in  an inteval between $ cdf_l = y_0+l_1*(x_1-x_0)$ and $cdf_u=y_0+l_2*(x_1-x_0)) $ such that $ max\{ |CDF(x_1)-cdf_l|,|CDF(x_1)-cdf_u|\}  \leq b_{err}$.}
%\footnote{too long sentence}
Since we want to scan as few tuples as possible during point queries execution, when constructing the tree, we can assign the $ b_{err} $ less than   half of the CDF capacity of a single data page. In this way, if a miss occurs on performing an index point query on a certain page, we only need to expand the search of its adjacent page.

\subsection{Implementations}

\label{Sec 4.2}

Directly implementing our model with the existing deep learning framework is inefficient, considering these frameworks tend to be designed for optimizing large models, resulting in significant invocation overhead which makes them infeasible for tasks like indexing~\cite{kraska2018case}. At the same time, the framework used in  RMI is designed for a fully connected network under single-dimensional indexing task, which is not compatible with our autoregressive network structure. Therefore, in order to accelerate the overall training and reasoning process of our structure, we adopted the following   implementation ideas. \looseness=-1

We first train our model using the PyTorch framework~\cite{AdamPaszke2019PyTorchAI} on GPU, then set up the tree structure in the Python environment~\cite{Python} and output it to disk.  The inference component implemented via C++ reads the trained structure into memory and uses the CPU for network reasoning to achieve the two tasks of CE and indexing.
\looseness=-1

%We first use pytorch framework to train our model on GPU, then we set up the structure of the tree in the python environment\footnote{url} and output it to  files. We use C++ to implement the inference environment. It firstly reads our CardIndex structure into the memory, and then uses CPU for network reasoning to realize the two tasks: cardinality estimation and indexing.

Our structure  is established with the following steps. First, we train the autoregressive model and use it as our root to calculate the Z-Order CDF values for all tuples. Then, we sort the data with the CDF value as the primary key. Given the number of second-level child nodes $ n $, the tuples with CDF values at $\left(\left(i-1\right)/n,i/n\right)$ are grouped into the $i$-th ($i<n$) linear sub-model. Finally, we maintain the maximum error bound $ b_{err} $ in linear refinement and calculate the slopes and intercepts of linear functions within each sub-model. 
\looseness=-1
%In the process of index establishment, we take \textcolor{red}{the following ideas:}\footnote{sketch idea?}
%First, we train the autoregressive model and make it as the root. We then use the model to calculate the Z-Order CDF values for all tuples. Then, the data is sorted with the CDF value as the primary key. Later, given the number of second-level child nodes $n$, the tuples with CDF values at $ ((k-1)/n,k/n) $ are grouped into the $ kth(k<n) $ linear subnode. Finally, we calculate the slope and intercept of the linear function in each child node. \textcolor{red}{We use gcc-Ofast as our compilation option. In addition, there are no other SIMD intrinsics   taken.}

\section{CardIndex Inference}

In this section, we attempt to use the trained CardIndex to solve the task of index query and cardinality estimation. We propose the processing algorithm of point and range query in \S~\ref{sec 5.1} and \S~\ref{sec 5.2}, respectively. The Naive Hybrid CE algorithm is proposed in \S~\ref{sec:CE}, while its improvements and properties are discussed in \S~\ref{sec FCE}.

\subsection{Point Query Processing}

\label{sec 5.1}

{In \S~\ref{sec 3.2} and \S~\ref{sec 3.3}, we derived that our CardIndex structure can simulate complex inner node traversal in the space partition tree by summing up the following conditional probability sequence:

\begin{equation}
	CDF(t) \simeq  \sum_{i \leq \ell} Pr(Z_{i}=0|Z_{1} \dots Z_{i-1} = z_1\dots z_{i-1}) \cdot z_{i}
\end{equation}

Considering that the parameter scale of our model is too tiny and we truncate the above sum from $n$ to $\ell$, there are errors in the above approximation. In \S~\ref{Sec 4.1}, we introduce the Linear Refinement strategy of using two sets of linear functions to refine the point query error. So that the CDF estimation of each record can be bounded between the interval sandwiched by two linear equations. Therefore, the flow of our algorithm is as follows. We first perform a CDF call to sum up the conditional probability sequence of the model. This replaces the complex node traversal process in the space partitioning tree of the traditional index structure. Then the obtained CDF value is refined through two linear functions to obtain the exact data page storing the particular record.

\begin{algorithm}[htb]
	\caption{ Point Index Query.}
	\label{alg:piq}
	\begin{algorithmic}[1] %这个1 表示每一行都显示数字
		\Require
		CardIndex Root, $M: Key \rightarrow CDF$;
		Linear submodel $S: CDF \rightarrow \{l_0,l_1,x,y\}$;
		Z-Order encoding of the queried tuple, $ Z = z_1 z_2 \dots z_n $;
		\Ensure
		The Address Pointer of the tuple: $Addr_t$;
		\State $ Acc = 1 $;
		\State $ CDF = 0 $;
		\State $ p_1,p_2\dots p_\ell = M(z_1,z_2,\dots z_\ell) $
		\For { $  i  =1,\dots  \ell $ }
		\State $ CDF = CDF + z_i \cdot (1-p_i) $;
		\State $ Acc  =Acc \cdot \left(\left(z_i \cdot p_i \right)+ \left(1-z_i \right) \cdot \left(1-p_i \right) \right) $;
		\EndFor
		\State $ l_0,l_1,x,y = S[CDF] $
		\State $ cdf_\ell =  l_0\cdot(Z-x) + y$
		\State $ cdf_u =   l_1\cdot(Z-x) + y$
		\State $ Page = FindPage(cdf_\ell,cdf_u) $
		\State $ Addr(t) = ScanPage(Page)  $
		\If{$ Addr(t) = NULL $}
		\State $ Addr(t) = AdjScanPage(Page)  $
		\EndIf
		\\ 
		\Return $Addr(t)$; %算法的返回值
	\end{algorithmic}
\end{algorithm}

To begin, the algorithm initially calculates the approximated CDF value in lines 1-6. It acquires the parameters of a linear function by utilizing the approximated CDF in line 7 and then uses them to locate the page where the tuple is stored(lines 8-10). Finally, the algorithm scans the data on the page in lines 11-14 to obtain the exact tuple address. If the tuple is not found on this page, the search is extended to adjacent pages. The complexity of the algorithm is $O(max( \ell, B))$, where $\ell$ denotes the inference length of our structure $M$, and $B$ is the size of the leaf page. Given that the page size and inference length are fixed, we can conclude that the above algorithm operates in constant time.

\subsection{Range Query Processing}

\label{sec 5.2}

\label{sec Range Query Processing}
The processing flow of our CardIndex structure for range queries resembles the range queries processing algorithm in Z-Ordered Index structures, such as UB-Tree\cite{ramsak2000integrating}. We obtain the two boundaries $Z_{min},Z_{max}$ of the query hyperrectangle $Q$ under Z-Order encoding. Utilizing the monotone ordering property in Z-Ordering, any tuple $t_i$ within the query box satisfies: $Addr(Z_{min}) \leq Addr(t_i) \leq Addr(Z_{max})$. Thus, we can obtain the page pointer \textit{ptrStart, ptrEnd} from the top-left and bottom-right corners by performing point queries twice, and then search within the sandwiched pages between pointers  \textit{ptrStart}  and  \textit{ptrEnd} .

However, many records between $Z_{min}$ and $Z_{max}$ may not  exist in the final result, as illustrated in Fig~\Ref{F4}. Therefore, blindly traversing all the pages between them will result in excessive scanning overhead. To address this issue, we utilize the  \textit{getBIGMIN}  method proposed by Tropf~\cite{tropf1981multidimensional}. If all the elements in the current scanning page are not within the query box, the  \textit{getBIGMIN}  method can obtain the next smallest Z-Order value contained in the actual query result page in $O(n)$ time, where $n$ is the Z-Ordered bit length. This feature allows us to skip incoherent pages in a Z-Order interval. The counterpart to the  \textit{getBIGMIN}  method is the  \textit{getLITMAX}  method, which obtains the previous largest Z-Order address containing the actual query results in $O(n)$ time. We will employ the  \textit{getLITMAX}  method in \S~\ref{sec FCE}. To conclude, the pseudocode description of our method is as follows:  \looseness=-1

\begin{algorithm}[htb]
	\caption{Range Index Query.}
	\label{alg:RIQ}
	\begin{algorithmic}[1] %这个1 表示每一行都显示数字
		\Require
		CardIndex Root, $M: Key \rightarrow CDF$;
		Linear submodel $S: CDF \rightarrow \{l_0,l_1,x,y\}$;
		QueryBox,$ Q $
		\Ensure
		The Result set: $R$;
		\State $ Z_{min},Z_{max} = ZEncoding(Q)$;
		\State $ ptrStart = PointQuery(M,S,Z_{min}) $;
		\State $ ptrEnd = PointQuery(M,S,Z_{max}) $;
		\State $ R = \phi $
		\State $ ptrSearch  =  ptrStart $
		\While {$ ptrSearch.zmax \leq ptrEnd.zmax  $}
		\State $ R_i = leafSearch(ptrSearch,Q) $
		\If{$ R_i = \phi $}
		\State $ NextZ = getBIGMIN(Q,ptrSearch.zmax)$
		\State $ ptrSearch = 	PointQuery(M,S,NextZ) $
		\Else
		\State $ R =R \cup  R_i $
		\State $ ptrSearch = ptrSearch.succ $
		\EndIf
		\EndWhile
		\\
		
		\Return $R$; %算法的返回值
	\end{algorithmic}
\end{algorithm}

In the first step of the algorithm, point queries are executed twice in lines 1-3 to obtain the pointers  \textit{ptrStart}  and  \textit{ptrEnd} , which correspond to the data pages containing the upper left and lower right endpoints of the query hyperrectangle. Then, in lines 6-13, the scanning pointer  \textit{ptrSearch}  traverses through the data space sandwiched between  \textit{ptrStart}  and  \textit{ptrEnd}, scanning the data pages one by one. If there is no intersection between the data page and the query box(line 8), the  \textit{getBIGMIN}  method is used to combine another point query in skipping to the next relevant page(lines 9-10). If an intersection exists between the page and the query box, the page is scanned to find all the elements in the query box, and  \textit{ptrSearch}  is updated to  \textit{ptrSearch.succ}, which points to the next adjacent page(lines 11-13).

We claim that the worst-case complexity of the range query algorithm above is $O\left(\left(B+n\right)\times\left(n_I\right)\right)$, where $B$ is the size of the page, and $n_I$ is the number of data pages intersecting the query box. This is because the length of the tuple after Z-Order encoding is $n$ bits,  \textit{getBIGMIN}  method with time complexity $O(n)$ is called $ n_I $  times,  and we scanned at most $ O(B\times n_I) $ tuples.

\subsection{Cardinality Estimation}
\label{sec:CE}

For cardinality estimation tasks, we can naturally use the proposed CardIndex as a cardinality estimator, as conditional probability correlation information between Z-Order bits of data is learned from  our autoregressive model.

 {Noting that the parameters of our model are so small. It can not achieve very accurate estimations of point densities. Direct engagement in estimating   small cardinality will bring huge estimation errors.} Therefore, we designed a hybrid estimation algorithm in achieving precise  estimation results within a much shorter time. In the case of large cardinality estimation, our estimation algorithm is similar to the progressive sampling in Naru\cite{yang2019deep}. Single inference through the PDF call of our model provides enough information for point density estimation. Based on the results of multiple progressive samples, we can estimate the cardinality of  {the queried region}. If the query selectivity is low enough, we directly use  our index structure to execute the query to obtain the exact cardinality. In this way, we can obtain precise cardinality values in a much shorter time. \looseness=-1

Our idea is outlined in Algorithm~\ref{alg:CE}. In line 2, we use the CDF estimation ability of the structure to probe a rough upper bound $ C_0 $ of the query cardinality(The correctness of the upper bound is proved in \S~\ref{sec FCE}). And $C_0$ is compared with a fixed threshold $ b^{EST} $ in line 3. In line 4, with the assistance of the index, we execute identified small cardinality queries to obtain the exact cardinality value. In lines 5-7, we use the progressive sampling algorithm for large cardinalities.

The implementation of rough estimation is straightforward. As illustrated in lines 9-13, we roughly suppose all the tuples whose Z-value is sandwiched between two endpoints $ Z_{min},Z_{max} $ are in the query box. We take $ cdf_L - cdf_U $ as our rough estimations.

As the attributes in Z-Order are bitwise entangled, the original progressive sampling algorithm cannot be applied directly. Thus, we modify the original progressive sampling strategy so that it can be efficiently adapted to CardIndex. Meanwhile, our pre-$ k $ neuron linkage provides us with the ability to ignore unnecessary computation by only focusing  on the $ k $ dependent neurons. Therefore, we modify the original progressive sampling method to enable quick execution under Z-Order (lines 15-25). In line 19 of our algorithm, we exploit the property of our structure's pre-$ k $ connection by using only the adjacent $k$ bits of the input to the current sampled bit string and reusing some intermediate results generated in the previous step to calculate the current probability information. This avoids iterating from the head of the sampling sequence each time, saving the overhead of repeated computations. Line 21 of our algorithm uses Z-encoded query box information to prune the obtained bit transfer path.

\begin{algorithm}[htb]
	\caption{ Hybrid Estimation.}
	\label{alg:CE}
	\begin{algorithmic}[1] %这个1 表示每一行都显示数字
		\Require
 
		CardIndex Structrue $M$;
		QueryBox,$ Q $ ;
		Sample Num: $ N_s $;
		\Ensure
		Cardinality Estimation of Query Box Q: $card$;
		
		\Procedure{CardEst }{$ Q,M ,N_s $}
		\State $ C_0 = FastCDFEst(M,Q) $  \Comment{Rough but fast est}
		\If{$ C_0 \leq b^{EST} $} \Comment{Using index to scan small card }
		\State $ card = RangeIndex(M,Q) $
		\Else
		\\
		\Comment{Using Progressive Sample to estimate big card }
		\State $ card = ZProgressiveSample(M,Q,N_s) $
		\EndIf
		\\
		\Return $ card $
		
		\EndProcedure
		
		\Procedure{FastCDFEst}{$ Q,M $}
		\Comment{Naive implement}
		\State $ Z_{min},Z_{max} = ZEncoding(Q) $
		\State $ cdf_L = M(Z_{min}) $
		\State $ cdf_U = M(Z_{max}) $
		\\
		\Return $ cdf_U -cdf_L$
		
		\EndProcedure
		
		\Procedure{ZProgressiveSample}{$ Q,M,N_s $}
		\Comment{ProgressiveSample under Z-Curve}
		\State $ \hat{P} =0 $
		\For{$ i=1\dots N_s $}
		\State $ p_i = 1, \textbf{s}=\textbf{0}_n $
		\State $ M.clear()  $ \Comment{Clear the intermediate result}
		\For{$ j = 1 \dots n $}
		\State $ \hat{P}(Z_i=1|\textbf{s}_{< i })= M(s_{i-k},\dots s_{i-1} ) $
		\State //Zero-Out the non-intercected bit branch
		\State $ \hat{P}(Z_i=1|\textbf{s}_{< i }) = prunBranch(Q,\textbf{s},j) $
		
		\State $ p_i =p_i \cdot \hat{P}(Z_i=1|\textbf{s}_{< i })  $
		\State $ s_i = sample\left( \hat{P}\left(Z_i=1|\textbf{s}_{< i }\right) \right) $
 
		\EndFor
		\State $ \hat{P}  =  \hat{P} +  p_i $
		\EndFor
		\\
		\Return $ (\hat{P} /N_s) \cdot N $
		
		\EndProcedure
	\end{algorithmic}
\end{algorithm}

Although the naive FastCDFEst algorithm may looks simple, it only requires two inferences. It is already effective for the detection of some small cardinality queries  since the locality of Z-Order makes the address of these low selectivity queries tend to be really close under bitwise  entangled encoding. However, there are also some corner cases, whose region is discontinuous under Z-Order coding(see Fig~\ref{F4}). 
This discontinuity, results in a very large estimation error of the naive FastCDFEst algorithm. Therefore, we improved the FastCDFEst algorithm, which alleviates the above problems. We discuss the improvement in \S~\ref{sec FCE}.

\subsection{Rough CDF Estimation Improvement}

\label{sec FCE}
In this section, we solve the problem left in the previous section. That is, for some corner case queries, their interval under Z-Order encoding is discontinuous. This large discontinuity leads to a large estimation error if we use naive FastCDFEst. Therefore, we develop an improved algorithm.

The intuition of our method is to conduct a binary search by  probing the query space under Z-Order. The pseudo-code of our algorithm is shown in Algorithm~\ref{alg:CER}. We recursively divide the search space by half and accumulate the interval values on both sides(lines 8-11). If the separation point is not in the query box, we use the means of calculating its $ bigmin $ and $ litmax $  values to correct the afterward search space, such the discontinuous region in the middle is skipped (lines 12-18).  \looseness=-1
	\begin{algorithm}[ htbp]
	\caption{ FastCDFEst-Refine.}
	\label{alg:CER}
	\begin{algorithmic}[1] %这个1 表示每一行都显示数字
		\Require
		Lower bound Z-value $ Z_{min}  $;
		Upper bound Z-value $ Z_{max}  $;
		QueryBox, $Q$;
		CardIndex Structrue $ M $;
		Current Recursion depth $ d $;
		\Ensure
		Refined Upper Bound of Q's Selectivity: $ sel $;
		
		\Procedure{FastCDFEst}{$ Z_{min} ,Z_{max} ,Q,M,d $}
		\State $ cdf_L = M(Z_{min}) $
		\State $ cdf_U = M(Z_{max}) $
		\State $ \delta =  cdf_U - cdf_L$
		\If{$ d \leq 0 $  or $ \delta \leq \delta_0$}  \Comment{Search space pruning}
		
		\State  $ sel = \delta $
		\Else \Comment{Binary CDF Search}
		\State $ Mid = (L+U)/2 $
		\If{$ M \in Q $}
		\State $ c_0 = FastCDFEst(Z_{min},Mid,Q,M,d-1) $
		\State $ c_1 = FastCDFEst(Mid,Z_{max},Q,M,d-1) $
		
		\Else
		\State $ bigmin = getBIGMIN(Q,Mid) $
		\State $ litmax = getLITMAX(Q,Mid) $
		\State \small $ c_0 = FastCDFEst(Z_{min},litmax,Q,M,d-1) $
		\State \small $ c_1 = FastCDFEst(bigmin,Z_{max},Q,M,d-1) $

		\EndIf
		\State $ sel = c_0 + c_1 $
		\EndIf
		\\
		
		\Return $ sel$
		
		\EndProcedure
	\end{algorithmic}
\end{algorithm}

In order to realize efficient computation in a limited time, we prune the search space using the following strategies. If the search depth exceeds budget $ (d \leq 0) $ or the search space's selectivity is too low  $ (\delta \leq \delta_0)$, we directly return the CDF distance between $ Z_{min} $ and $ Z_{max}  $(lines 1-7). In the worst case, the time complexity of the above algorithm is $ O(n \times 2^{d_0}) $, where $ d_0 $ is the maximum search depth, and $ n $ is the Z-Ordered binary bit encoding length.

We use Fig~\ref{F4} as an example to illustrate the above process. For the query $  (X = 0 , 4 \leq Y \leq 8)$ in Fig~\ref{F4}, direct use of all tuples between query endpoints 16-64 as an approximation will bring great errors. On the contrary, we take the middle point 40 and calculate its $ bigmin $  21 and $ litmax $ 64. We then use them to narrow the search space. We can observe that, through a single recursive search step, we can reduce the rough cardinality estimation of the query region from 48 to 7, which is reduced to its original   $ 1/7 $.
			
\begin{figure}[htbp] 
	\centering
	\includegraphics[height=6.5cm]{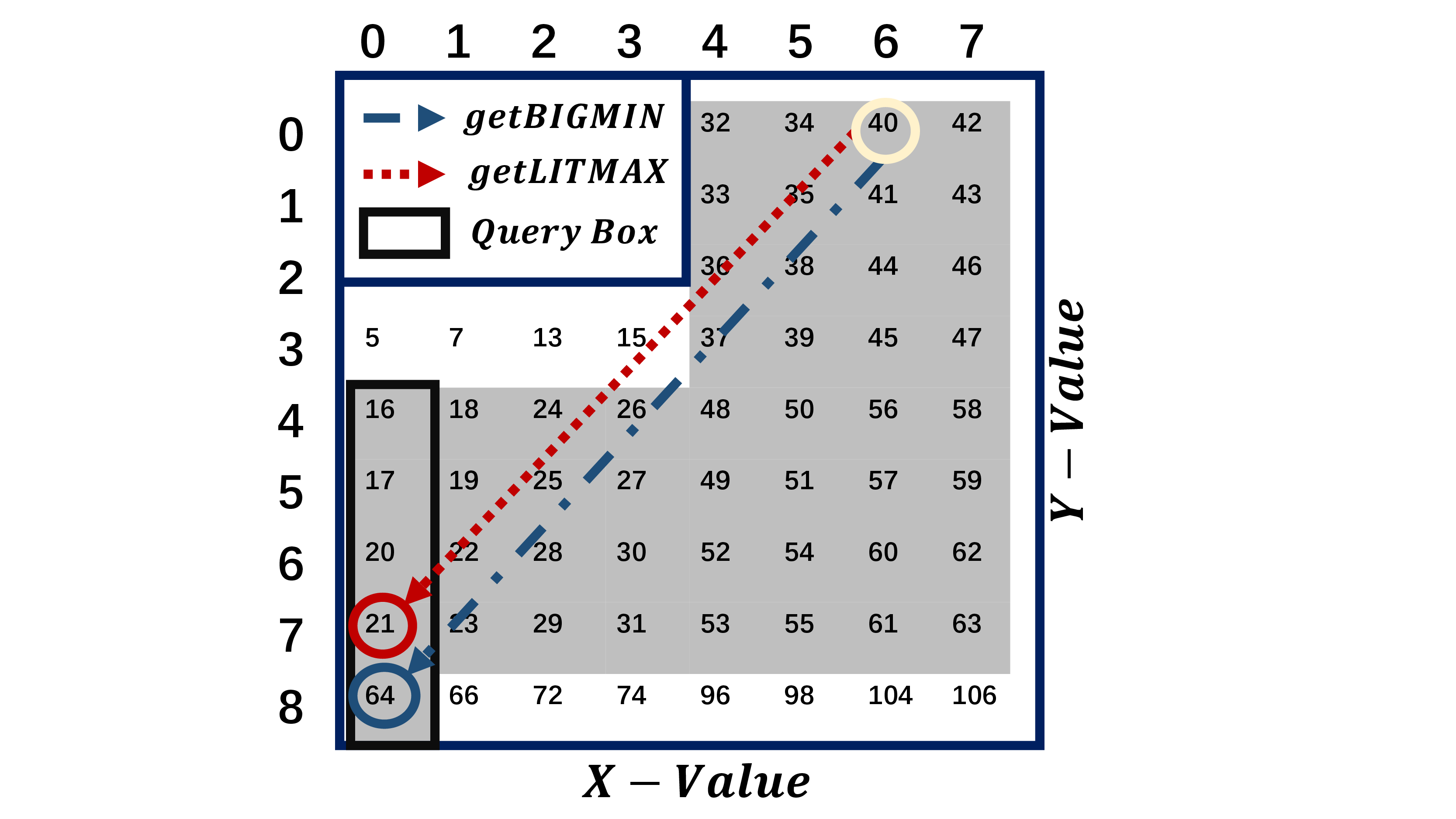}
	\caption{Refining FastCDFEst}
	\label{F4}
\end{figure}

A natural question is, does our proposed hybrid algorithm get caught up in a trap  searching too many tuples, resulting in performance degradation?  {We claim that this will never happen. The number of the tuples we scanned $ N_{scan} $ is limited under our preset selectivity threshold $ b^{EST} $. The correctness is proved through the following theorem.}

\begin{theorem}
	$ N_{scan}\leq b^{EST}\times  N$, where $ N_{scan} $ is the number of the tuples scanned, $ N $ is the cardinality of the whole table, and $ b^{EST}  $ is our preset selectivity threshold in Alg~\ref{alg:CE} .
\end{theorem}
\begin{proof}
	(Sketch): 
	Since we only perform index execution under the condition that the estimated cost $ C_0\leq b^{EST} $, the direction of our  proof is to prove that $ N_{scan}\leq C_0\times  N$.
	In \S~\ref{sec Range Query Processing}, we analyzed that the number of tuples scanned $ N_{scan} $ is smaller than $  B\times n_I $, where $ n_I $ is the number of pages intersecting the query box, and $ B $ is the size of the page. According to this conclusion, we give the intuition that the results of both $ FastCDFEst $ algorithms designed will be larger than this value. 
	
	For naive $ FastCDFEst $, {many Z-Ordered pages between $ Z_{min} $ and $ Z_{max} $ are not intersected by the query box. Therefore, $  B\times n_I \leq (cdf_U -cdf_L)\times N$. Given that $ N_{scan} \leq  B\times n_I $ and $ C_0 =  (cdf_U -cdf_L) $ in naive $ FastCDFEst $, we conclude that $ N_{scan}\leq C_0 \times N $.}
	
	For improved $ FastCDFEst $,  its recursive search, in fact, can be thought of as a pre-execution of limited   \textit{getBIGMIN}  calls in a range query. Due to its limited recursively search depth, it cannot remove too many discontinuous intervals, and thus the  ``jumped" distance should be less than the discontinuous Z-value interval of the  ``skip" in real execution, which means that $ ((cdf_U -cdf_L)\times N - C_0 \times  N ) \leq ((cdf_U -cdf_L)\times N - B\times n_I )   $. Therefore, $  B\times n_I \leq C_0\times N$ is obtained, which derives that $ N_{scan}\leq C_0 \times  N $. 
	
	In summary, both of our algorithms will maintain $ N_{scan}\leq C_0 \times  N $. Since we only execute $ C_0 \leq b^{EST} $ queries, we derived that  $ N_{scan}\leq  b^{EST}  \times  N $.
\end{proof}

The above upper bound ensures that the maximum number of tuples we scanned is limited to be smaller than $ b^{EST} \times  N$ when we perform Alg~\ref{alg:CE}'s  hybrid method in obtaining the exact value of cardinality. Index scans are performed only when the scanning cost is lower than $ b^{EST} \times  N$. Therefore, we just set the maximum index scanning selectivity   $ b^{EST} $  beforehand. Our hybrid estimation strategy will never fall into the trap of performing index scans on larger cardinalities than $ b^{EST} \times  N $.

\section{Experiments}

In this section, we answer the following questions experimentally.

1. Compared with the state-of-the-art selectivity estimators, how does CardIndex perform in terms of the estimated quality of different workloads, model size, training time, and inference time? (\S~\ref{sec5.2})

2. Compared with the state-of-the-art traditional multidimensional  indexes, how does CardIndex perform in terms of range query, point query, model size, and construction time? (\S~\ref{sec5.3})

3. How do the techniques adopted in the approach, i.e. Pre-$ k $ neuron linkage, FastCDFEst, affect the quality and efficiency of the CardIndex structure? (\S~\ref{sec5.4})

\subsection{Experimental Setup}

\textbf{Environment.} The experiments are done on a computer with an AMD Ryzen 7 5800H CPU, NVIDIA  RTX3060 GPU, 64 GB RAM, and a 1 TB hard disk. We use PyTorch in Python to train our model on GPU and use C++ on CPU to infer our model on two tasks: Index and Estimation. For the baselines of cardinality estimation, we implemented them in Python, and for the baselines of indexing tasks, we implemented them in a C++ environment.

\textbf{Competitors.} We choose the following baselines in Index and Estimation.

(1) \textit{Naru}~\cite{yang2019deep}: The state-of-the-art estimator on single table's cardinality estimation. Using the  autoregressive model and progressive sampling to estimate the query region's cardinality. We used the default net parameters in its methods

(2) \textit{Sample}~: Sample several records in memory for CE. The sampled Ratio is set as 1\%.

(3) \textit{DeepDB}~\cite{hilprecht2019deepdb}: Use Sum-Product Network to estimate cardinality.

(4)\textit{ MHIST}~\cite{ViswanathPoosala1996ImprovedHF}: The multidimensional histogram is used to store the PDF approximation and to predict the cardinality value based on it.

(5)\textit{ KDB-Tree}~\cite{robinson1981kdb}: This index uses a KD-tree with a B-tree structure to support block storage. We adopt the state-of-the-art implementation from Moin's baseline code in the Waffle project \cite{moti2022waffle}.

(6) \textit{R-Tree}:\cite{beckmann1990r} The state-of-the-art traditional multidimensional index. We adopt implementation from boost.org.

\textbf{Datasets.} We use three real-world datasets and one synthetic dataset for experimental study on index and estimation tasks.

(1)\textit{Power}~\cite{Power}, an electric power consumption data, which owns a large domain size in all attributes (each $ \approx $ 2M). Our snapshot contains 2,049,280 tuples with 6 columns.

(2) \textit{DMV-INT}~\cite{DMV}, Real-world dataset consists of vehicle registration information in New York. We use the following 11 columns with widely differing data types and domain sizes (the numbers in parentheses): record type (4), reg class(75), state (89), county (63), body type (59), fuel type (9), valid date (2101), color (225), sco ind (2), sus ind (2), rev ind(2). Given that there are many non-numeric data, we sort each column and convert each string of it to its corresponding integer coordinate value. Our snapshot contains 12,300,116 tuples with 11 columns.

(3)\textit{OSM}\cite{OSM}, Real-world geographical data. We downloaded the data set   of Central America from OpenSteetMarket and selected two columns of latitude and longitude. This data set has 80M tuples(1.8GB in size). Our snapshot contains 80,000,000 tuples with 2 columns.  \looseness=-1

(4)\textit{Synthetic}, Synthetic geographical data. We made a slight change to the OSM data and scaled it to 128M by copying its first 48M rows. Our snapshot contains 128,000,000 tuples with 2 columns.

\textbf{Evaluation Metrics.} To evaluate the performance on index and estimation tasks, we adopt the following perspectives: \textit{Q-error metric}: defined as $ Q(E,T) = $  $ max\{\frac{E}{T},\frac{T}{E}\} $ where $ E $ is the estimated cardinality value and  $ T $ is the real cardinality. We reported the entire Q-error distribution as  (50\%, 95\%, 99\%, and  Q-Max quantile) of each workload. \textit{Latency metric} : for index tasks, we report their average response time for each query execution. For estimation tasks, we report the average inference time for each estimation.  \textit{Size metric} :
We report the total size of the non-leaf nodes for the index structure and the size of the overall structure for an estimator model.

\textbf{Parameter Settings.} We adopt the original settings of all the baseline methods. For  the CE tasks of our  CardIndex(CI) in \S~\ref{sec5.2}, we fixed the Minimum estimation bound $b^{EST}$  in our method as $ 1e^{-2} $ and Linked Neuron's number as 32.  For our CI's index performance study in \S~\ref{sec5.3}, we allocate the number of linear sub-models through the following principle{: For each 1M tuple, we assign 128 sub-models. Therefore, we  allocate 256 sub-models on the Power data set, 1536 on the DMV data set, and 10272 on the OSM data set.}  The linear error bound for the CardIndex's Linear Refinement is set to $ 1/10 $ of the page CDF capacity by default, and the default minimum pruning precision of $ \delta_0 $ of
our $ FastCDFEst $ is set to 0.001.  The fanout number is set to 100 for all index experiments including CI's submodels and other baseline methods.

\textbf{Workloads.} We generate multidimensional  queries through the following procedures. Firstly, we randomly select several rows from the data as the center of the query, select an integer from the range of $ [1, D] $ as the number of predicates, and then sample a number from $ (0, s_i\cdot w_i) $ among the selected columns as the width of the query, where $ s_i  $ is the scaling coefficient to generate some small cardinality queries, and $ w_i $ is the distance between the maximum and minimum values of the column. We first generate several queries under the three $ s_i$  $  (1,0.01,0.001)$  and then mix them. Then, according to their selection degree, they are divided into   three groups with each group having 1k queries, i.e. High($ sel \geq 1e^{-2} $),  High($ sel \geq 1e^{-4} $), Extreame-Low ($ sel \leq 1e^{-4} $).   (See Fig.\ref{wlD})

\begin{figure}[]
	\centering
	\includegraphics[height=5cm]{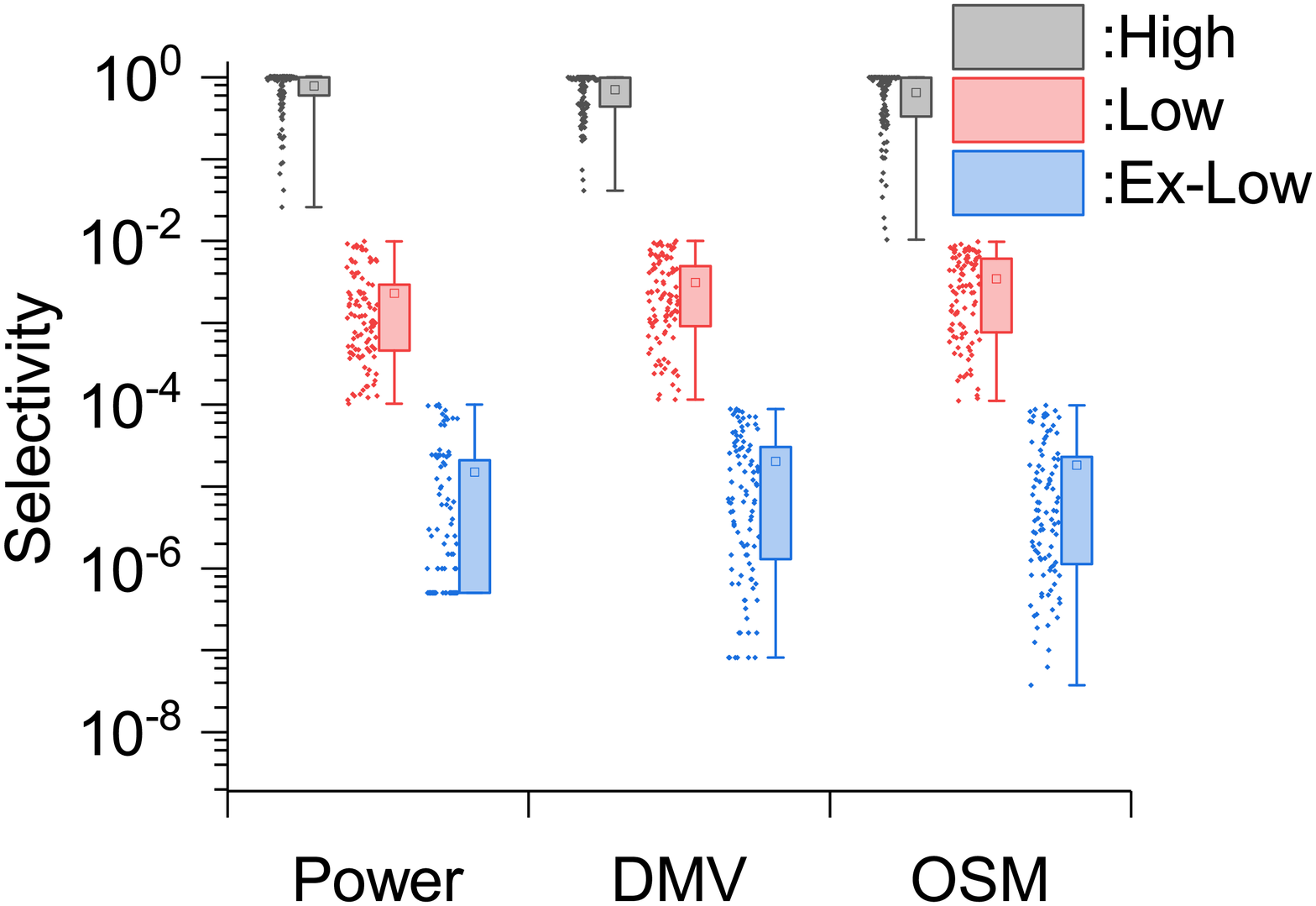}
	\caption{Distribution of workload selectivity \\ \centering{(Sampled 10\% from the query workload for demo)}}
	\label{wlD}
\end{figure}
\subsection{Estimation Evaluation}
\label{sec5.2}

In general, in terms of estimation quality on large cardinalities workloads, our method is slightly worse (2.5$ \times $) than the current SOTA method due to the parameter scale of our model being 2-3 orders of magnitude smaller than these models. However, the estimation quality on low selectivity workloads is  1.3-114 $\times $ better than the current SOTA data-driven cardinality estimation strategy. In terms of the estimated latency, although there is a slight distance between our method with the direct  sampling strategy, the average estimation latency of our method on CPU is up to 12 times faster than the inference of Naru's method on GPU. In terms of training time, under the same epoch iteration, our network training is 2 orders of magnitude faster than Naru's.

The details are as follows.  We tested the existing baseline estimators on three different workloads in three real-world datasets. The number of the progressive samples in our Card Index (CI for short) is fixed as 2,000 as Naru's.  

\begin{table*}[ht]\label{TabQE}
	\caption{Q-errors,Avg latency(ms) on 3 real-world Datasets's 3 workloads}
	\setlength{\tabcolsep}{1.6mm}{
		\begin{tabular}{|c|c|cccc|cccc|cccc|l|}
			\hline
			\multirow{2}{*}{\textbf{DataSet}} & \multirow{2}{*}{\textbf{Estimator}} & \multicolumn{4}{c|}{\textbf{High($ sel \geq 1e^{-2} $)}} & \multicolumn{4}{c|}{\textbf{LOW($ sel \geq 1e^{-4} $)}} & \multicolumn{4}{c|}{\textbf{Extreame Low( $ sel \leq 1e^{-4} $)}} & \multicolumn{1}{c|}{\multirow{2}{*}{Time(ms)}} \\ \cline{3-14}
			&  & \multicolumn{1}{c|}{50th} & \multicolumn{1}{c|}{95th} & \multicolumn{1}{c|}{99th} & \multicolumn{1}{c|}{MAX} & \multicolumn{1}{c|}{50th} & \multicolumn{1}{c|}{95th} & \multicolumn{1}{c|}{99th} & \multicolumn{1}{c|}{MAX} & \multicolumn{1}{c|}{50th} & \multicolumn{1}{c|}{95th} & \multicolumn{1}{c|}{99th} & MAX & \multicolumn{1}{c|}{} \\ \hline
			\multirow{6}{*}{Power} & Naru &  {1.00} & {1.04} & {1.12} & 1.43 & 1.10 & 1.58 & 2.20 & 5.96 & 1.18 & 3.00 & 5.01 & \textbf{8.0} & 17.4 \\ \cline{2-2}
			& Sample & 1.00 &  {1.02} &  {1.03} &  \textbf{{1.03}} & 1.47 & 671 & $ 1e^3 $ & $ 1e^3 $ & 14.0  & 157  & 194  & 201  & \textbf{2.09} \\ \cline{2-2}
			& DeepDB & \textbf{1.00} & \textbf{ {1.01}} &  \textbf{{1.02}} &  1.18 & 1.08  & 2.44 & 5.23 & 18.8 & 1.30 & 4.67 & 10  & 37.0  & 12.98 \\ \cline{2-2}
			& MHist & 1.00 &  3.13 &  4.90 &  10.1 & 3.35 & 289 & $ 1e^3 $ & $ 9e^3 $ & 2.91  & 67.1  & 130  & 204  & 1086 \\ \cline{2-2}
			
			& CI(d=0) & 1.00 & 1.06 & 1.27 & 1.89 &  {1.24} &  {2.20} &  {2.93} &  {34.7} &  {1.00} &  {1.41} &  {5.00} &  {70.0} & 22.6\\ \cline{2-2}
			& CI(d=16) & 1.00 & 1.06 & 1.27 & 1.89 & \textbf{1.00} & \textbf{1.06} &\textbf{ 1.31} & \textbf{2.12} & \textbf{1.00} &\textbf{ 1.00} &\textbf{ 1.00 }&  {11.0}  &17.08 \\ \hline
			\multirow{6}{*}{OSM} & Naru & 1.00 & 1.08 & 1.13 & 1.28 & 1.18 & 2.91 & 6.79 & 432 & 2.81 &  42.2 &  141 & 311 & 84.1 \\ \cline{2-2}
			& Sample & \textbf{1.00} & \textbf{1.01} &  {1.02} & \textbf{1.02} &  {1.22} & 2.24 & 293 & $2e^4$ & 238 & $ 4e^3 $ & $ 7e^3 $ & $ 7e^3 $ &\textbf{1.09 }\\ \cline{2-2}
			& DeepDB & 1.00 &  {1.01} &  \textbf{1.01} &  1.03 & 1.24  & 1.93 & 3.57 & 11.7  & 2.89   & 56.9  & 254  & $ 2e^{3} $& 7.08 \\ \cline{2-2}
			& MHist & 1.08 &  1.57 &  2.37 &  4.15 & 2.67 & 26.1 & 51.7 & 159 & 15.3  & $ 1e^{3} $  & $ 2e^{3} $ & $ 5e^3 $  & 5235 \\ \cline{2-2}
			
			& CI(d=0) & 1.00 & 1.22 & 1.79 & 2.69 & 1.00 & 2.22 & 3.05 & 30.6 & 1.00 & 1.36 & 2.03 & 21.0 & 5.65 \\ \cline{2-2}
			& CI(d=16) & 1.00 & 1.22 & 1.79 & 2.69 & \textbf{1.00 }& \textbf{1.24} &\textbf{ 1.70} & \textbf{10.9} & \textbf{1.00} & \textbf{1.00} &\textbf{ 1.35} & \textbf{8.94 }& 7.06 \\ \hline
			\multirow{6}{*}{DMV} & Naru & 1.00 & 1.04 & 1.06 & \textbf{1.09} & 1.90 & 3.56 & 47.3 & 943 & 2.98 & 4.3  & 55.9 & 250 & 30.9 \\ \cline{2-2}
			& Sample & \textbf{1.00} & \textbf{{1.02}} & \textbf{1.06} & 2.59 &  {1.14} &  {1.79} & 17.5  & $ 1e^3 $ & 46.5 & 708 & 874 & 969 & \textbf{4.14} \\ \cline{2-2}
			& DeepDB & 1.92 &  2.11 &  2.66&  33.7 & 1.98 & 8.18 &  131 & $ 6e^{3} $& 2.18  & 29.8 & 220  & 908  & 6.78\\ \cline{2-2}
			& MHist & 1.93 &  2.62 &  3.20 &  35.6 & 3.57 & 19.4 & 50.1 & 685 & 13.0  & 174  & 399  & 686  & 1055 \\ \cline{2-2}
			
			&  CI(d=0)  & 1.00 & 1.08 & 1.32 & 2.75 & 1.42 & 3.00 &  {3.97} &  {36.6} &  {2.43} &  {8.22} &  {17.0} &  {147} & 11.5 \\ \cline{2-2}
			& CI(d=16)& 1.00 & 1.08 & 1.32 & 2.75 & \textbf{1.00} &\textbf{ 1.71} &\textbf{ 2.40} & \textbf{8.21} & \textbf{1.00 }&\textbf{ 1.00}&\textbf{ 1.00}&\textbf{ 10.0}& 7.6 \\ \hline
	\end{tabular}}
\end{table*}

As can be seen from Table 1, for large cardinality estimation,   existing cardinality estimators all perform well. Compared to the state-of-the-art methods, our method performs slightly worse on this type of load, where the maximum Q-error is twice as high as Naru's. This is because, under these high selectivity workloads, our hybrid estimation algorithm cannot use the index to enhance the performance on the low selectivity query. Therefore, the method degenerates into a progressive sampling algorithm under tiny models. Considering that our network has the smallest number of parameters, only 10-100KB of parameter size (1/100 - 1/1000 of Naru's), the small parameters of the model make it impossible for us to use a single structure to learn detailed fine-grained distribution information like Naru's, thus bringing  errors to the estimation.
When the selectivity of the query decreases, the advantage of our method is demonstrated via hybrid estimation. At Low selectivity workloads($ sel \geq 1e^{-4} $), the maximum Q-error is outperformed by $ 2$-$114\times $ compared with Naru's and $ 125$-$2000 \times $ compared with sampling. On the Ex-Low workload, the performance distance is also pronounced, i.e. the $ 1.3$-$38  \times$ smaller  Q-error against Naru's and $ 3$-$222\times $ smaller than DeepDB's. We analyzed the reasons for our baseline's poor performance as follows.

Under such low selectivity workload, it is often hard for Sample to extract elements from the queried region, so it tends to obtain the prediction result of 0 cards, resulting in  Q-error being the size of the real cardinality number. For DeepDB, it is difficult to accurately delineate correlations between columns, therefore for OSM and DMV data, which may have a strong correlation between attributes, it produces a large error. Although Naru can use autoregressive models to learn correlations between attributes, it is still very difficult to obtain accurate estimations with such extremely-small cardinality. Because such a small cardinality actually imposes a great demand on the point density prediction accuracy and the number of progressive samples of the model. The limited progressive-sampling points in Naru's setting can hardly cover such a complex space intercepted by multiple range predicates. Therefore, insufficient sampling and low precision point estimation result in large errors in  prediction on some corner cases with small cardinalities. To conclude, for these above data-driven methods, it is impossible to learn completely precise distribution information under a limited parameter storage budget, and thus for these small cardinality estimations, there are often very large errors.

In contrast,  with the help of an auxiliary index structure, we can efficiently search through the original table and retrieve the precise results of these small cardinality queries. Thus, the advantage of our hybrid estimation is that the estimations of small cardinalities are performed efficiently directly using the index to obtain a true and error-free distribution. Although the method is efficient enough to be effective for 50\% of Ex-low workload queries at $ d=0 $ (the naive implementation mentioned in \S~\ref{sec:CE}), deeper search steps of   \textit{FastCDFEst}   refine the search space, making our CardIndex's classification more precise and easier to identify small cardinality queries executed on the index. When the search depth increases from 0 to 16, the maximum Q-error of the mixed estimation is reduced by 10 times. The \textit{FastCDFEst} method is extremely effective, as under Power and DMV dataset, 99\% of the extremely Low workloads have been accurately identified and precisely estimated, with a Q-error of 1.0  by index execution. 

At the same time, in terms of latency, it is worth acknowledging that the time efficiency of sampling is the best. Given that we do not have to perform that much computation on such methods, going through the sampled data is  the only cost. Meanwhile, it should be noted that although CardIndex uses progressive sampling as Naru does, the inference latency of progressive sampling of ours performed on CPU is even 1.02-12 times faster than that of Naru method implemented on GPU. This is for two reasons. On the one hand, our network's parameters are 1/100 of the size of Naru's, so inference does not require as much computation. On the other hand, for small cardinality queries, direct  execution on the fly will be much faster than entirely sampling estimation algorithms. Meanwhile, as our experimental results show, a larger depth does not necessarily lead to more time consumption. Meanwhile, since the direct index execution time of a small cardinality is usually 2-3 orders of magnitude smaller than these estimation algorithms, the larger the depth is, and the faster the estimated time on Power and DMV data sets is.\looseness=-1

We also report the training time and the structure size of the estimation model in Fig~\ref{FigConT} and Fig \ref{FigConS}. We fixed the epoch of training to 20 in both neuro-based methods. It is noted that the training time of our network is 1-2 orders of magnitude faster than that of Naru due to the network parameters of our models(CI-Core) being 1/100 of Naru's (\ref{FigConT}).

\begin{figure}[htbp]
	\centering
	\subfigure[  Construction Time]{
		\label{FigConT}
		\includegraphics[width=0.45\columnwidth]{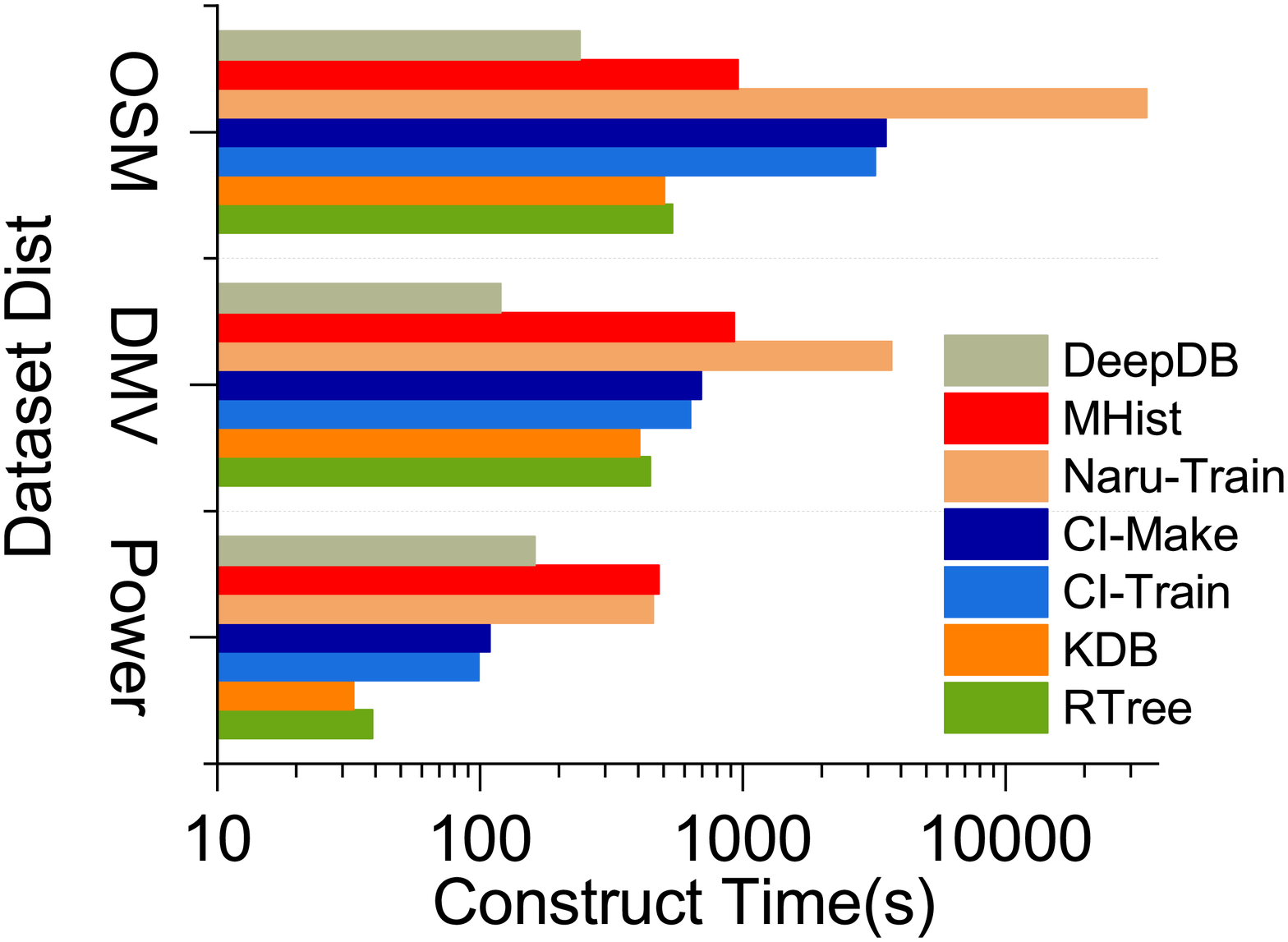  }
	}
	\subfigure[  Structure Size]{
		\label{FigConS}
		\includegraphics[width=0.45\columnwidth]{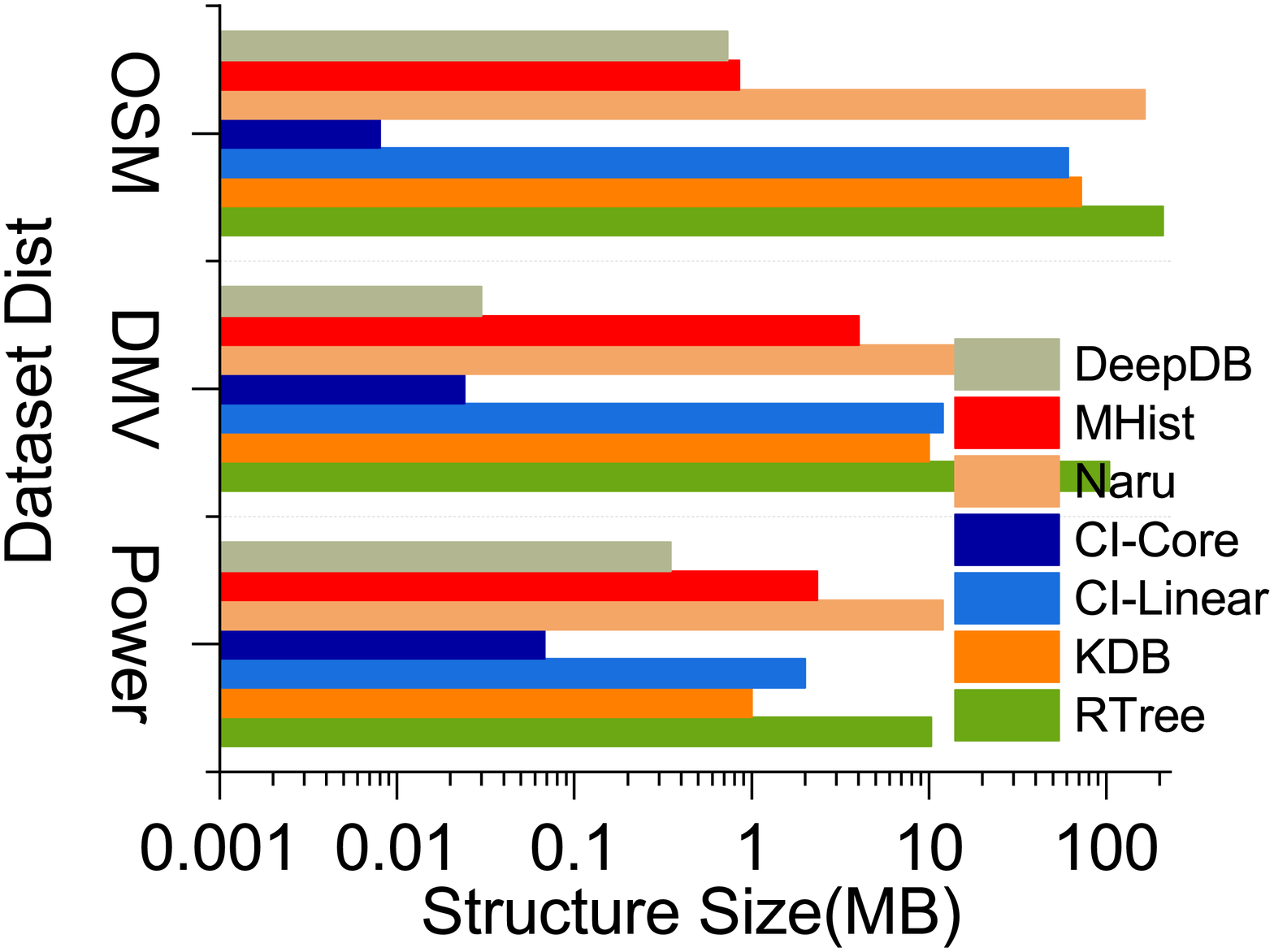}
	}
	\subfigure[  Construction  Scalability]{
		\label{FigConTSCA}
		\includegraphics[width=0.45\columnwidth]{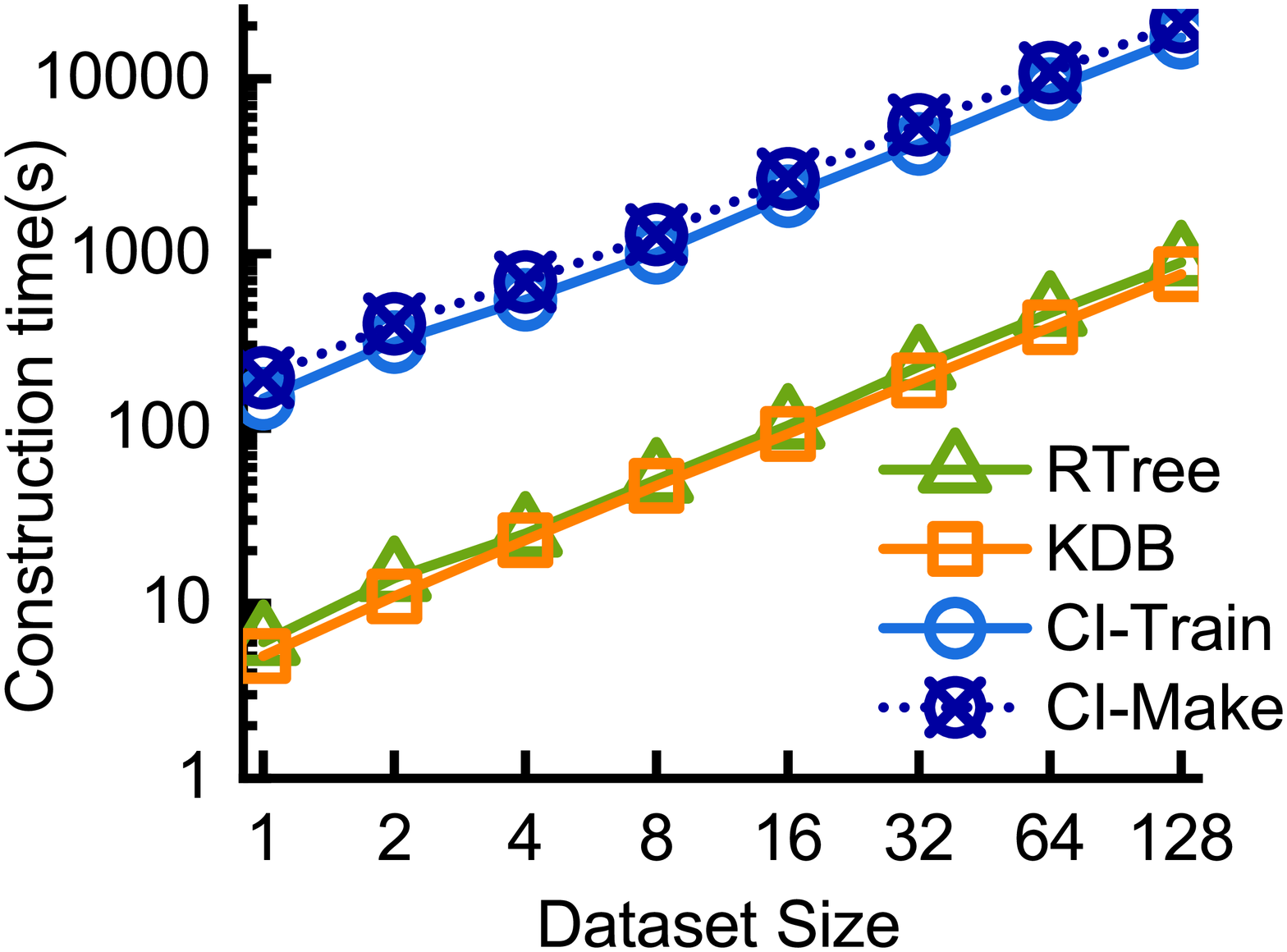  }
	}
	\subfigure[    Size Scalability]{
		\label{FigConSSCA}
		\includegraphics[width=0.45\columnwidth]{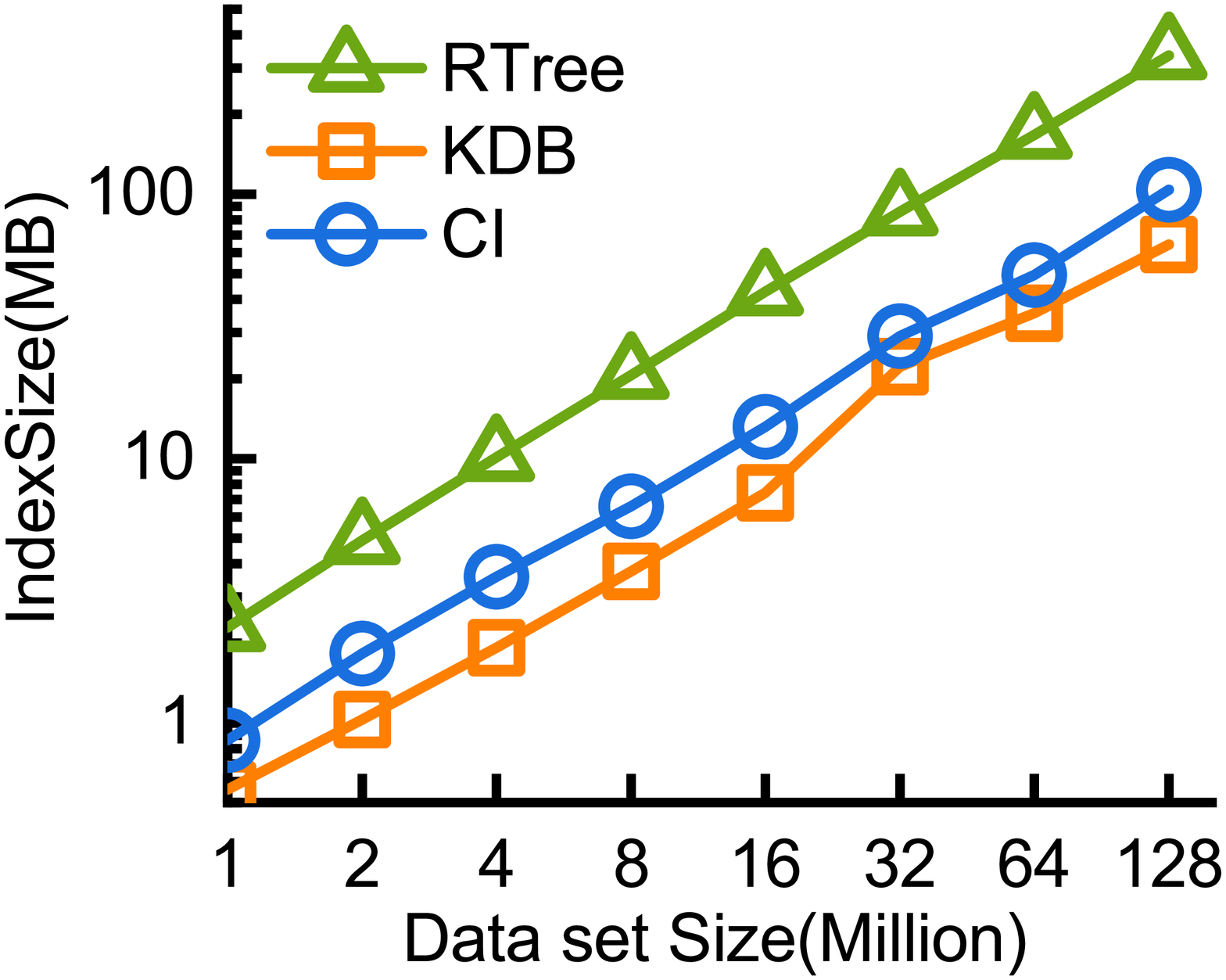}
	}

	\caption{Structure Construction Info }
	\label{Exp8}
	\vspace{-1.5em}    % Give a unique label
\end{figure}
		
\subsection{Index Evaluation}
\label{sec5.3}

\textbf{Structure Construction.} We report the size and construction time of our structure and  baseline approaches  taken on Fig~\ref{Exp8}. It is not difficult to point out that, in terms of space, our structure occupies a space similar to the memory space of a traditional index compared to the current mainstream traditional methods. Also in terms of building time, compared to other deep autoregressive approaches, our approach is much faster.  (CI's construction is 2 orders of magnitude faster than Naru's). However, there is still room for improvement in our approach compared to traditional index building. (The construction time of CI is 10 times slower than KDB and RTree).
%\textbf{Point Query.} We sample 1k points from each dataset, use them as query points and report the average response time.  As Fig~\ref{FigPQ-real} shows, our CI has the best point performance on all datasets. It is $ 30\% - 40\% $  faster than the SOTA traditional structures on three real data sets. The results in the synthetic data set are consistent with those in the real data set. For example, the point query on the OSM dataset takes $ 4.4\mu s $ for KDB,  $2.9 \mu s $ for the R tree, and  $ 2.2  \mu s$ for our CardIndex. Given the fact that our structure is also a learned index, this experimental fact is not surprising. Because our structure can calculate the value of a given tuple coordinate in constant time, it avoids the log-scale scanning traversal of blocks in the traditional structure.

\textbf{Point Query.} We sample 1k points from each dataset, use them as query points and report the average response time.  As Fig~\ref{FigPQ-real} shows, our CI has the best point performance on all datasets. It is $ 30\% - 40\% $  faster than the SOTA traditional structures on three real data sets. The results in the synthetic data set are consistent with those in the real data set. For example, the point query on the OSM dataset takes $ 4.4\mu s $ for KDB,  $2.9 \mu s $ for the R tree, and  $ 2.2  \mu s$ for our CardIndex. Given the fact that our structure is also a learned index, this experimental fact is not surprising. Because our structure can calculate the value of a given tuple coordinate in constant time, it avoids the log-scale scanning traversal of blocks in the traditional structure. \looseness=-1

\begin{figure}[thbp]
	\centering
	\subfigure[  Real-World Dist]{
		\label{FigPQ-real}
		\includegraphics[width=0.45\columnwidth]{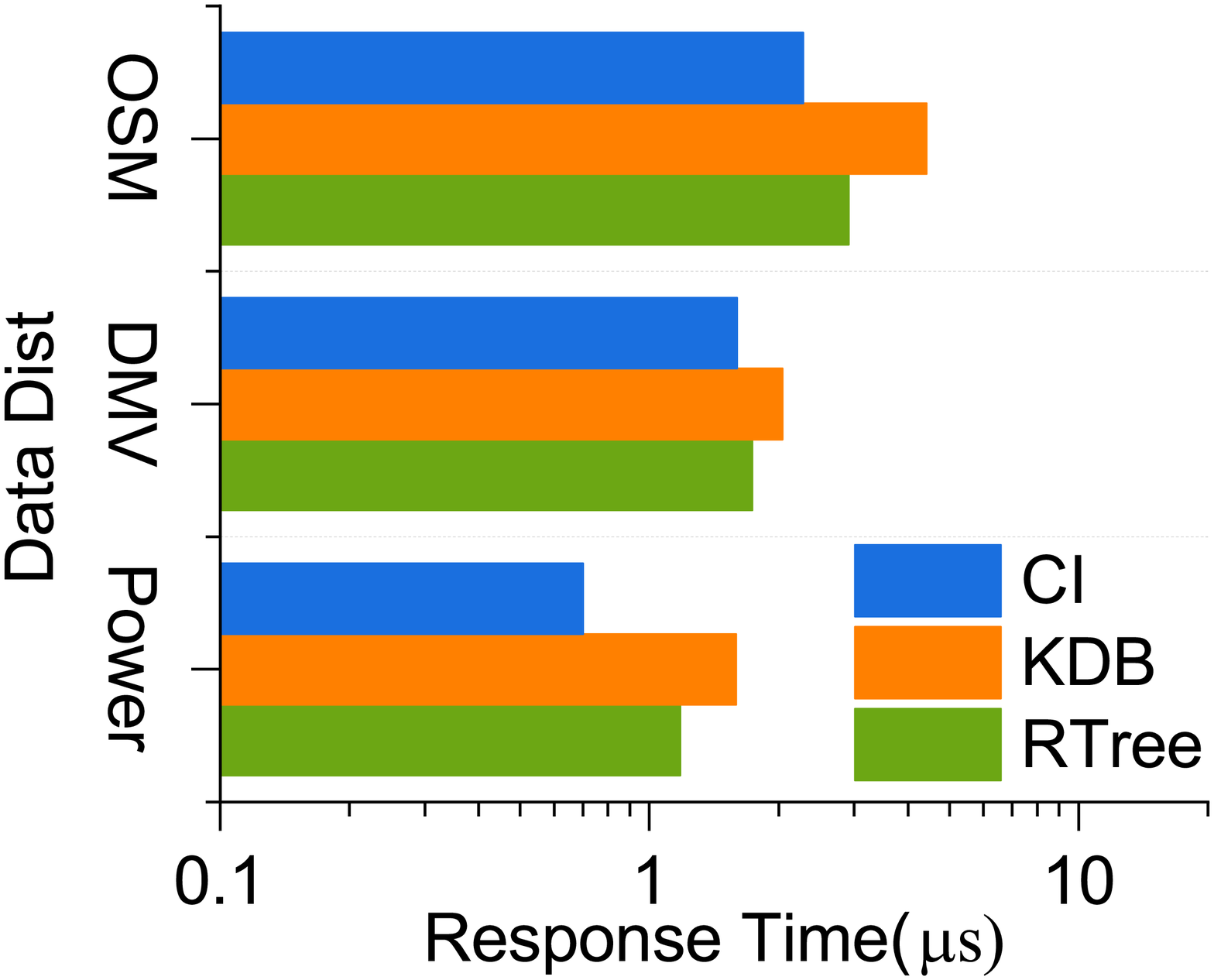  }}
	\subfigure[   Query Scalability ]{
		\label{FigPQ-sca}
		\includegraphics[width=0.45\columnwidth]{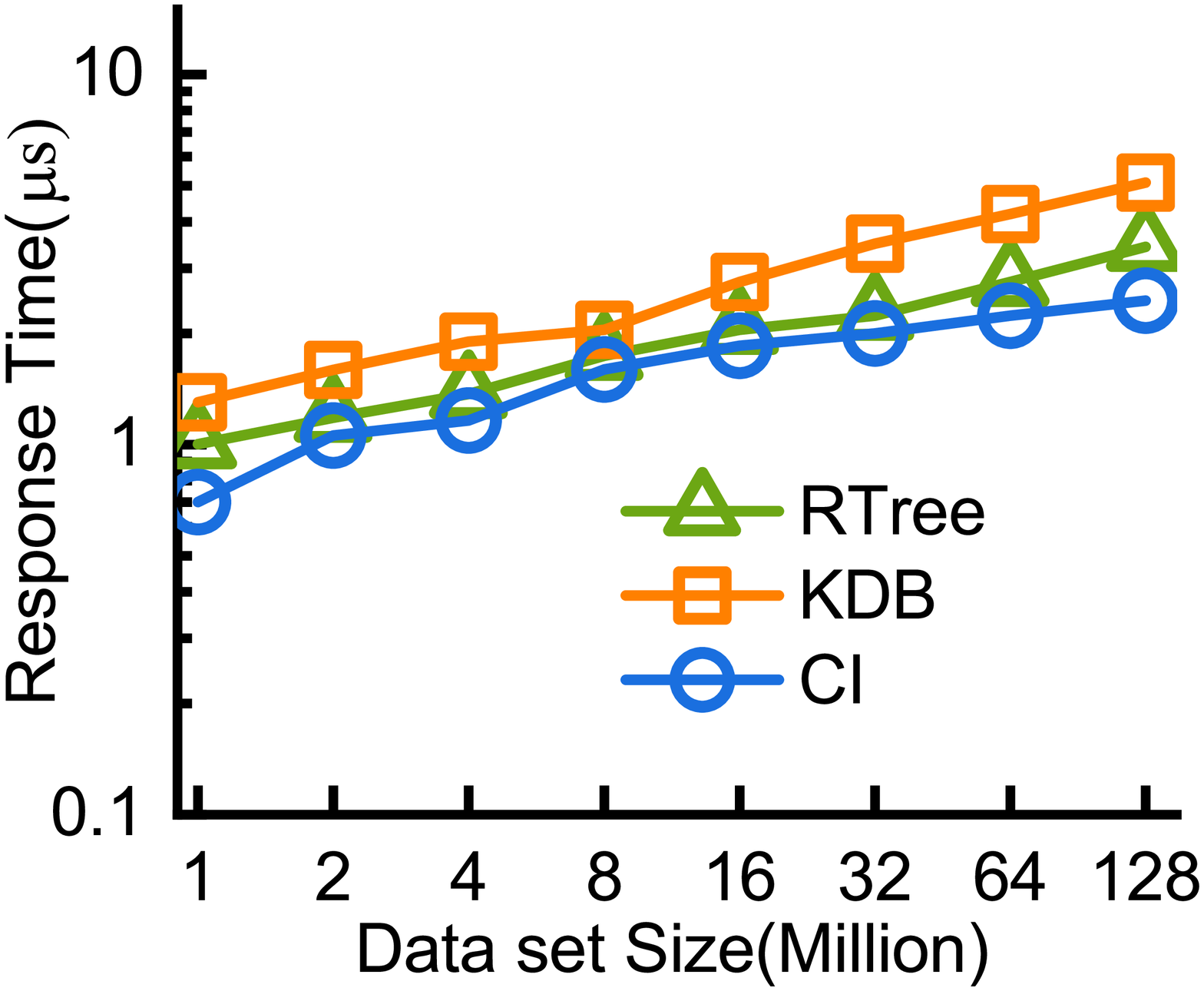}}
	
	\caption{Point Query Details }
	\label{Expx9}
	\vspace{-1.5em}    % Give a unique label
\end{figure}

\textbf{Range Query}. We directly evaluate the index task using two query workloads generated when evaluating the cardinality estimator. We chose High(\textit{W1}) and Low workloads(\textit{W2}). As under ExtremeLow workload, the result set is too small, and the result is quite similar to the point query. For the interests of space, we did not show the corresponding experimental results. We find that the average time of several structures is similar under High workloads(\textit{W1}). While on Low workloads(\textit{W2}), our CardIndex is 4-10 times faster than KDB and RTree. Indeed, RTree is twice as fast as our approach on the Power data set. Our debugging observations show that this problem is related to the ordering of the untuned Z-Order curve may not be optimal for some certain range queries. The average selection among the accessed blocks is so low that only a small fraction of each fetched block actually falls into the query area. Tuning Z-Order and a more compact block could solve this problem. We will solve this problem in future work.

\begin{figure}[thbp]
	\centering
	\subfigure[  Real-World Dist]{
		\label{FigRQ-real}
		\includegraphics[width=0.45\columnwidth]{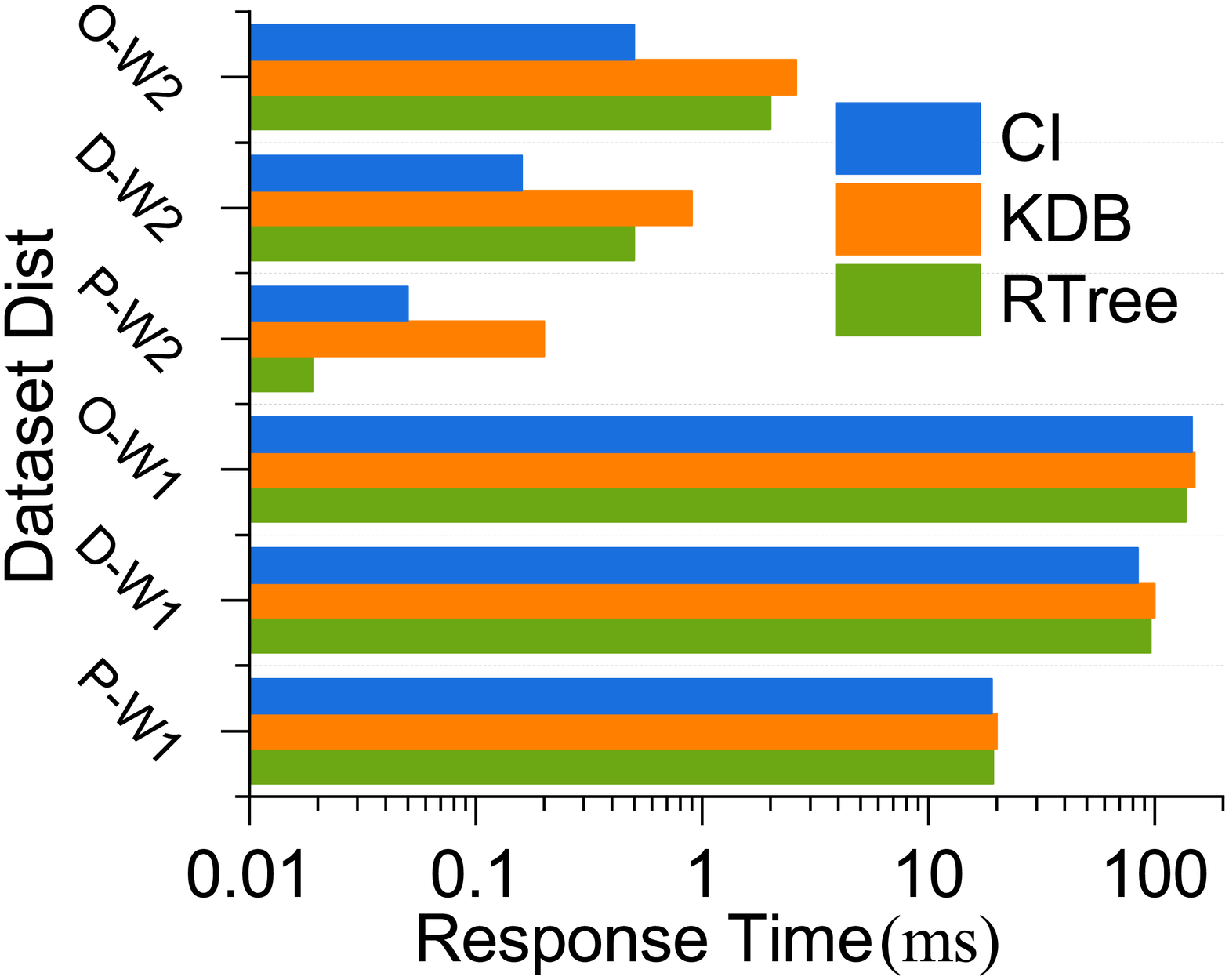  }}
	\subfigure[ Query Scalability ]{
		\label{FigRQ-sca}
		\includegraphics[width=0.45\columnwidth]{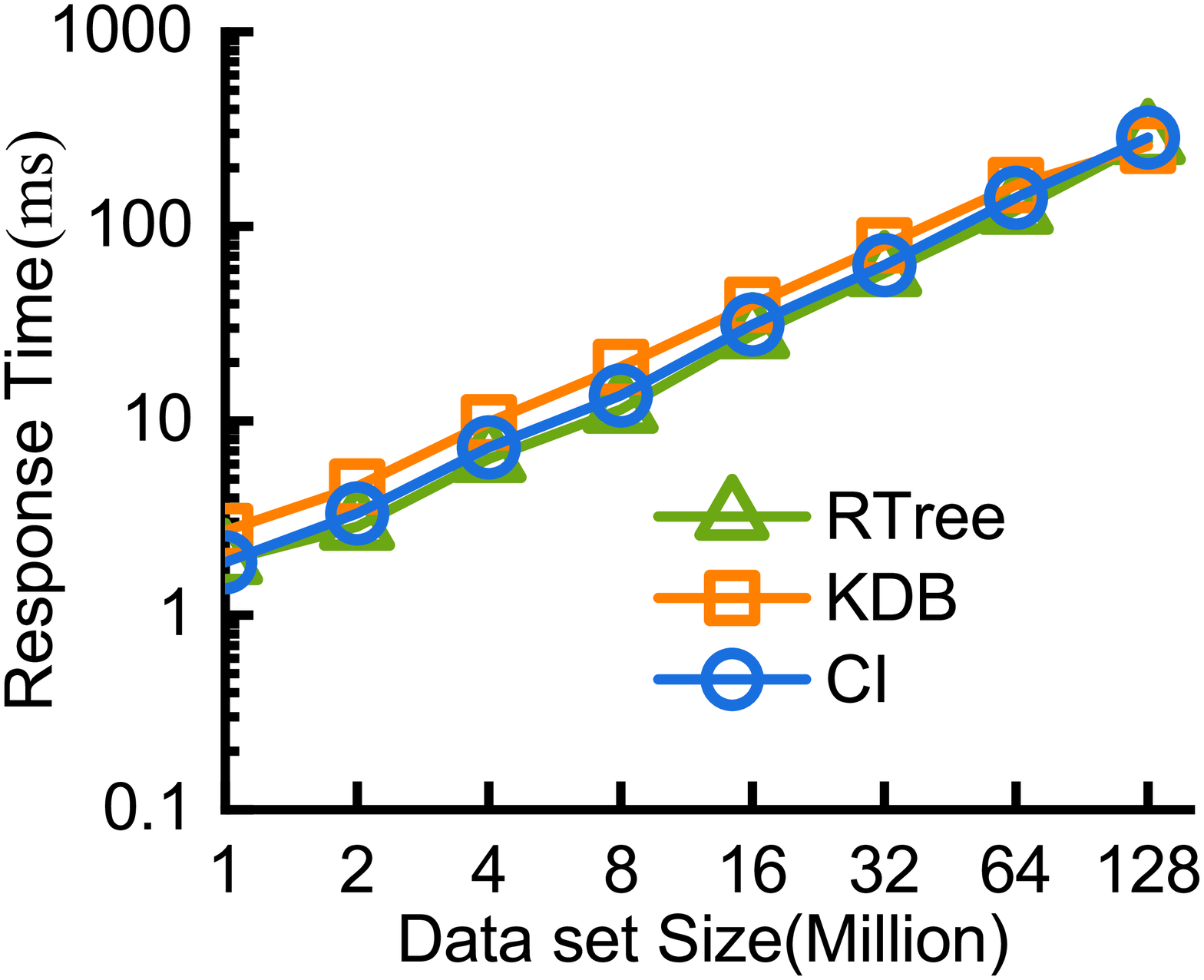}}
	
	\caption{Range Query Details }
	\label{ExpxRQ}
	\vspace{-1.5em}    % Give a unique label
\end{figure}

To sum up, our CardIndex is  $ 30\% $-$ 40\% $  faster on point query processing task than SOTA on three real data sets. For range queries, the performance is similar to that of R trees in general cases. And under specific workloads, it can be 4 to 10 times faster than the traditional structure. For range queries, the performance is similar to that of R trees. And under specific workloads, it can be 4 to 10 times faster than the traditional structure

\subsection{Variance Evaluation}

%这部分是不是消融实验？可以和心玥等讨论之
%是的,我咨询了一下学姐，然后参考了下CE文献里的命名方式，把这块改名
In this section, we evaluate the effects of techniques employed in our approach: pre-$ k $ neuron Linkage and FastCDFEst with a certain depth. We use  Ex-Low workload on the Power dataset.

\begin{figure}[thb ]
	\centering
	\subfigure[  Inference Time]{
		\label{F9A}
		\includegraphics[width=0.45\columnwidth]{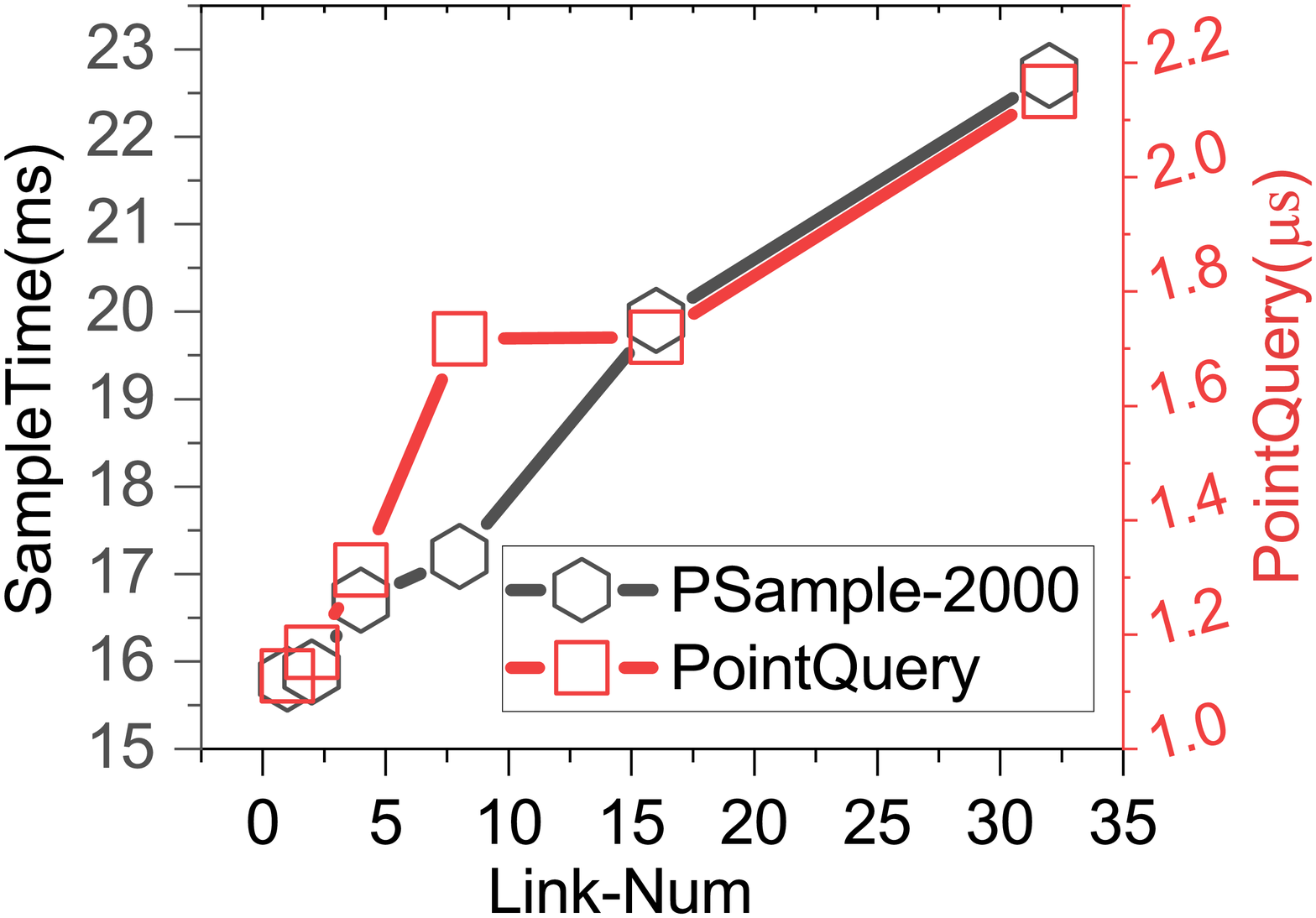 }}
	\subfigure[  top-10\% Q-error]{
		\label{F9b}
		\includegraphics[width=0.45\columnwidth]{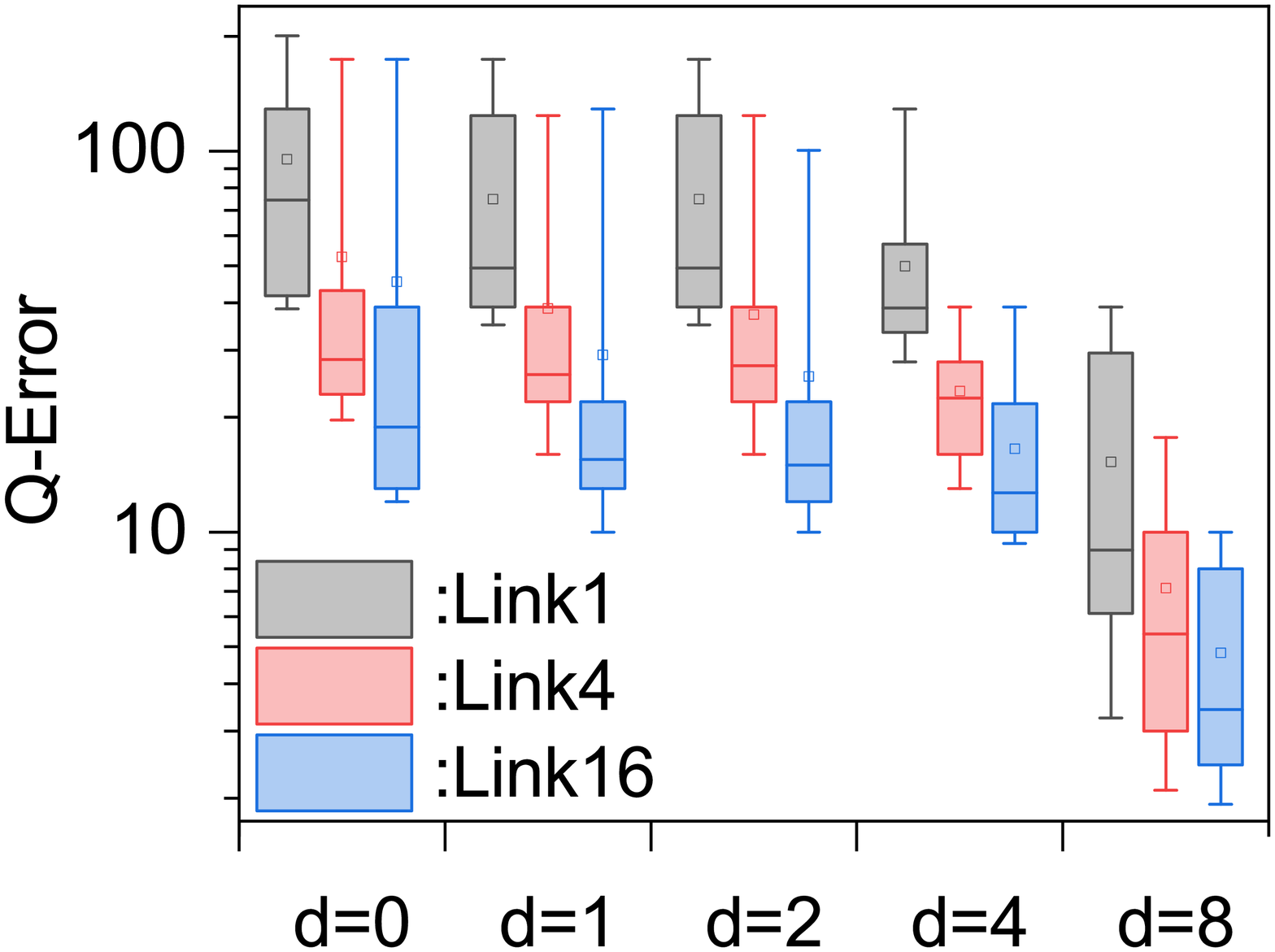 }}
	
	\caption{Link }
	\label{Exp9}
	\vspace{-1.5em}    % Give a unique label
\end{figure}

\textbf{Effect On Linkage Number.} In Fig~\ref{F9A},  we vary the Link-Num and report the average estimation time and point Query time of our method. It can be observed that as the number of links increases, more computation is required, so the progressive sampling time and point query time get longer. The benefits brought by increasing the number of connections are intuitive from Fig~\ref{F9b}. When the number of connections is increased, the overall Q-error is reduced in all search depths. This is because,  more connections of neurons shall bring a more  powerful learned network, which makes the result of progressive sampling more accurate.

\textbf{Effect On FastCDFEst Search Depth.}   We also report the effect of  the search depth $ d $ on the quality of the estimation problem in Fig~\ref{F9b}. It is clear that deeper search depth leads to better estimation quality. For example, as search depth $ d $ increases from 0 to 8, the maximum Q-error drops by almost an order of magnitude, even in the case of Link-Number equals 1. The reason is that a deeper CDF search depth will locate small cardinality queries easier. Therefore, more estimation queries are delegated to the index than inaccurate progressive sampling, making Q-error lower. 

\label{sec5.4}
\section{Conclusions}

In this paper, we propose CardIndex, a lightweight multidimensional learned index with cardinality support, which killed two birds with one stone.  These birds that our stone killed, not only in the literal sense, i.e. use a structure to solve both indexing and CE tasks.  It is more in the obstacles we solve. (1) Our CardIndex fixed the duplication distribution storage waste. (2) Our CardIndex solved the dilemma between the lightweight model deployment and corner-case accuracy. Our experiment validated that not only does our method's index performance exceed traditional structures $ 30\%-40\% $ in point query task, 4-10$ \times $  in range query task, but also outperforms the state-of-the-art CE methods by 1.3-114$ \times $   with 1-2 orders of magnitude smaller storage and computational overhead.

%\end{document}  % This is where a 'short' article might terminate

\bibliographystyle{abbrv}
\bibliography{reference}  % vldb_sample.bib is the name of the Bibliography in this case

%Generated by bibtex from your ~.bib file.  Run latex,
%then bibtex, then latex twice (to resolve references).

%APPENDIX is optional.
% ****************** APPENDIX **************************************
% Example of an appendix; typically would start on a new page
%pagebreak

\end{document}